\documentclass[12pt]{amsart}
\usepackage{amsmath, amsthm, amssymb, amscd, hyperref}
\usepackage{arydshln} 
\usepackage[headings]{fullpage}
\usepackage{xypic}
\xyoption{all}

\usepackage{mathrsfs} 
\usepackage{wasysym} 

\usepackage{color}   

\newcommand*\colvec[3][]{
    \begin{pmatrix}\ifx\relax#1\relax\else#1\\\fi#2\\#3\end{pmatrix}
} 

\makeatletter 
\let\@wraptoccontribs\wraptoccontribs 
\makeatother

\numberwithin{equation}{section}
 
\theoremstyle{plain}
\newtheorem{theorem}[equation]{Theorem}
\newtheorem{corollary}[equation]{Corollary}
\newtheorem{lemma}[equation]{Lemma}
\newtheorem{proposition}[equation]{Proposition}
\newtheorem{question}[equation]{Question}

\theoremstyle{definition}
\newtheorem{definition}[equation]{Definition}
\newtheorem{definitions}[equation]{Definitions}

\newtheorem{remark}[equation]{Remark}

 \numberwithin{figure}{section}

\DeclareMathOperator{\Spec}{Spec}
\DeclareMathOperator{\LSpec}{LSpec}

\newcommand{\cat}{\mathsf}

\DeclareMathOperator{\Ring}{\cat{Ring}}
\DeclareMathOperator{\pRing}{\cat{pRing}}
\DeclareMathOperator{\pBRing}{\cat{pBRing}}
\DeclareMathOperator{\cRing}{\cat{cRing}}

\DeclareMathOperator{\Bool}{\cat{Bool}}
\DeclareMathOperator{\pBool}{\cat{pBool}}
\DeclareMathOperator{\fBool}{\cat{fBool}}
\DeclareMathOperator{\Set}{\cat{Set}}

\DeclareMathOperator{\Top}{\cat{Top}}
\DeclareMathOperator{\Frm}{\cat{Frm}}
\DeclareMathOperator{\Loc}{\cat{Loc}}

\DeclareMathOperator{\AffSch}{\cat{AffSch}}

\DeclareMathOperator{\pSpec}{\mathit{p}-Spec}
\DeclareMathOperator{\KS}{KS}

\DeclareMathOperator{\Proj}{Proj}
\DeclareMathOperator{\Idpt}{Idpt}
\DeclareMathOperator{\Span}{Span}
\DeclareMathOperator{\range}{range}

\DeclareMathOperator{\Hom}{Hom}
\DeclareMathOperator{\pt}{pt}

\DeclareMathOperator{\RIdl}{RIdl}

\newcommand{\Ringop}{\Ring\op}
\newcommand{\cRingop}{\cRing\op}
 
\newcommand{\C}{\mathbb{C}}
\newcommand{\Q}{\mathbb{Q}}
\newcommand{\Z}{\mathbb{Z}}
\newcommand{\F}{\mathbb{F}}
\newcommand{\R}{\mathbb{R}}
\newcommand{\M}{\mathbb{M}}
\newcommand{\GL}{\mathrm{GL}}

\newcommand{\setS}{\mathcal{S}}
\newcommand{\p}{\mathfrak{p}}

\renewcommand{\O}{\mathcal{O}}
\newcommand{\catC}{\mathcal{C}}

\newcommand{\op}{^\mathrm{op}}
\newcommand{\comm}{\odot}

\newcommand{\sym}{\mathrm{sym}}

\newcommand{\two}{\mathbf{2}}

\newcommand{\separate}{\bigskip}

\begin{document}

\title[A Kochen-Specker theorem for integer matrices]
{A Kochen-Specker theorem for integer matrices\\
and noncommutative spectrum functors}

\author{Michael Ben-Zvi}
\address{Tufts University\\
Department of Mathematics\\
Bromfield-Pearson Hall\\
503 Boston Avenue\\
Medford, MA 02155, USA}
\email{michael.ben\_zvi@tufts.edu}

%

\author{Alexander Ma} 
\address{Department of Mathematics\\
University of Minnesota\\
206 Church St.\ SE\\
Minneapolis, MN, 55413, USA}
\email{maxx0234@umn.edu}

\author{Manuel Reyes}
\address{Department of Mathematics\\
Bowdoin College\\
8600 College Station\\
Brunswick, ME 04011--8486, USA}
\email{reyes@bowdoin.edu}
\urladdr{http://www.bowdoin.edu/~reyes/}
\thanks{This material is based upon work supported by the National Science Foundation 
under grant no.\ DMS-1407152}

\contrib[with Appendix by]{Alexandru Chirvasitu}
\address{Department of Mathematics\\
Box 354350\\
University of Washington\\
Seattle, WA 98195, USA}
\email{chirva@math.washington.edu}

\date{August 11, 2017}

\begin{abstract}
We investigate the possibility of constructing Kochen-Specker uncolorable sets of idempotent 
matrices whose entries lie in various rings, including the rational numbers, the integers, and finite 
fields. Most notably, we show that there is no Kochen-Specker coloring of the $n \times n$ idempotent 
integer matrices for $n \geq 3$, thereby illustrating that Kochen-Specker contextuality is an inherent 
feature of pure matrix algebra.
We apply this to generalize recent no-go results on noncommutative spectrum functors, 
showing that any contravariant functor from rings to sets (respectively, topological spaces or locales)
that restricts to the Zariski prime spectrum functor for commutative rings must assign the 
empty set (respectively, empty space or locale) to the matrix ring $\M_n(R)$ for any integer 
$n \geq 3$ and any ring $R$.
An appendix by Alexandru Chirvasitu shows that Kochen-Specker colorings of idempotents in 
partial subalgebras of $\M_3(F)$ for a perfect field $F$ can be extended to partial algebra
morphisms into the algebraic closure of $F$.
\end{abstract}

\maketitle

\section{Introduction}
\label{sec:intro}

The Bell-Kochen-Specker Theorem~\cite{Bell, KochenSpecker} is a no-go theorem that demonstrates 
the impossibility of certain hidden variable theories for quantum mechanics. 
The usual formulation of the Heisenberg Uncertainty Principle in terms of matrix mechanics
shows that we can only expect to have precise knowledge of the values of two quantum-mechanical
observables $P$ and $Q$ simultaneously if these observables (represented as operators on some Hilbert
space) commute: $PQ = QP$. Thus commuting observables are also called \emph{commeasurable}. The 
nature of the algebra of operators on a Hilbert space $H$ (of dimension $\dim(H) \geq 2$) is such that
one may have an observable $P$ commeasurable with two observables $Q$ and $Q'$, but such that
$Q$ and $Q'$ are not commeasurable with one another. Thus one may expect to have precise
simultaneous knowledge of the values of $P$ and $Q$, or of $P$ and $Q'$, but not of the triple
$\{P,Q,Q'\}$. 
A hidden variable theory is called \emph{non-contextual} if the value $v(P)$ assigned to an observable
$P$ is independent of the choice of pairwise commeasurable set of observables $\{P,Q_1,Q_2,\dots\}$ 
that also happen to be measured by an experimental setup. 
This property was emphasized by Bell in~\cite[Section~V]{Bell}. 

Now suppose that one restricts to observables that are projections (i.e., self-adjoint idempotent
operators) on $H$. The value of each projection, being an eigenvalue of the operator, is either~$0$
or~$1$, so that such observables represent ``yes-no questions'' that may be asked about the
underlying quantum system. 
Further, if one has an orthogonal set $\{P_i\}$ of projections whose sum is the identity (such as 
the projections onto an orthonormal basis of $H$), these classically correspond to mutually exclusive,
collectively exhaustive propositions about the system. If one measures the values of the $P_i$
simultaneously, then compatibility with the classical logic of Boolean algebras would require that one 
of these projections is assigned the value~$1$ and the rest are assigned value~$0$. Thus, a 
non-contextual hidden variable theory is expected to ``color'' every projection on $H$ with a
value~$0$ or~$1$ in such a way that, for each basis $\{v_i\}$ of $H$, the projection onto exactly one 
of the $v_i$ is assigned the value~$1$. But Kochen and Specker proved such an assignment to be
impossible whenever $\dim(H) \geq 3$, by providing an explicit set of orthogonal projections
(represented by vectors in their ranges) for which such a $\{0,1\}$-valued function does not exist.
Bell~\cite{Bell} provided an alternative proof using Gleason's 
Theorem. (Our own methods follow closely those of Kochen and Specker, especially considering
colorings of finite sets of vectors or idempotents. For this reason, we will refer to the
no-hidden-variables theorem as the \emph{Kochen-Specker Theorem}.)

By now there are many fine discussions of the role of the Bell-Kochen-Specker theorem in the logical 
foundations of quantum mechanics, so we have limited our discussion of this background to a 
brief explanation of the physical intuition behind the mathematical result. Aside from the original
papers of Bell and Kochen-Specker, we refer readers to the textbooks~\cite{Belinfante, Laloe} for
introductions to the theorem in the broader context of hidden variable theories, to the
article~\cite{BarrettKent} for a discussion of various claimed ``loopholes'' to the theorem, and to 
the detailed survey~\cite{StanfordEncyclopedia} for further discussion and references to the literature.

There is much recent interest in examining the Kochen-Specker Theorem from new perspectives.
One of the most notable such programs is the formulation of the theorem in the context of
topos theory~\cite{IshamButterfield}. 
There are also approaches to the general theory of contextuality through sheaf 
theory~\cite{AbramskyBrandenburger} and through graphs and hypergraphs~\cite{CWS, AFLS}.
There has been recent progress~\cite{UijlenWesterbaan} on the problem of determining
lower bounds for the size of a Kochen-Specker uncolorable set of three-dimensional vectors.
The theorem has also found recent application in the well-known ``Free Will Theorems'' of Conway 
and Kochen~\cite{ConwayKochen1, ConwayKochen2}.

Our present work seeks to push the study of the Kochen-Specker Theorem in a new direction by 
allowing the study of contextuality for vectors and matrices whose entries lie in more general coefficient 
rings than the real or complex numbers, and it is motivated by applications in the setting of
noncommutative algebraic geometry.
The original analysis of Kochen and Specker framed the discussion of hidden variables in algebraic
terms as an assignment of values to all observables on a quantum system whose restriction to 
any commesaurable set of observables forms a \emph{homomorphism}. Such an assignment of values
will be called a \emph{morphism of partial rings}, as discussed in more detail in Section~\ref{sec:partial}
below. 
From this perspective, one may view any noncommutative ring $R$ as a purely algebraic analogue of the 
observables of a quantum system, with its commutative subrings as ``commeasurable'' subsets of
observables, so that a morphism of partial rings from $R$ to a commutative ring can be viewed 
as a ``noncontextual hidden variable theory.''
	
In this paper we establish that contextuality---in the purely algebraic sense of inadmissibility of
such morphisms of partial rings---is a property
inherent to any matrix ring of the form $\M_n(R)$ for $n \geq 3$, independent of the choice of 
the ring of scalars $R$. We consider this problem from the intimately related
perspectives of Kochen-Specker colorings of idempotents and of morphisms of partial rings. 
Section~\ref{sec:partial} contains background and fundamental results on partial rings and partial
Boolean algebras in the sense of Kochen and Specker, showing the precise relationships between 
colorability of idempotents, morphisms of partial rings, and the spectrum $\pSpec(R)$ of 
\emph{prime partial ideals} of a partial ring $R$. Most of these relationships are expressed in 
the basic language of categories~\cite{MacLane}, in terms of various functors and natural 
transformations. 
The results in this section are elementary but seem not to have been carefully considered 
elsewhere; thus we hope that this will fill a gap in the literature. 

Then in Section~\ref{sec:coloring} we prove that algebraic analogues of the Kochen-Specker 
theorem do or do not hold in various (partial) rings of matrices. 
These results are summarized in Table~\ref{table:main results}. Given a ring $S$, we let 
$\Idpt(S)$ denote the set of idempotents of $S$, which carries the structure of a partial
Boolean algebra. A formal definition of a Kochen-Specker coloring is given in 
Definition~\ref{def:KS coloring}. For a commutative ring $R$, the set of matrices in 
$\M_3(R)$ that are symmetric (equal to their own transpose) is denoted $\M_3(R)_\sym$.
We remark that the partial rings $S$ in Table~\ref{table:main results} for which 
$\pSpec(S) = \varnothing$ admit no morphism of partial rings $S \to C$ for any (total) 
commutative ring $C$, yielding a direct analogue of the type of obstruction that Kochen 
and Specker originally sought. 

Our motivation for this study stems from the recent application of the Kochen-Specker theorem
to noncommutative geometry in~\cite{Reyes}. There it was shown that any contravariant functor
$F$ from rings to sets (or to topological spaces) whose restriction to commutative rings is the 
prime spectrum functor $\Spec$ must satisfy $F(\M_n(\C)) = \varnothing$ for $n \geq 3$. In
Section~\ref{sec:spectrum} we strengthen this result to conclude that such functors $F$
in fact satisfy $F(\M_n(R)) = \varnothing$ for every ring $R$ and integer $n \geq 3$.

\begin{table}
\centering
\caption{Idempotent colorings and partial spectra of partial rings of matrices} 
\label{table:main results}
\begingroup
\renewcommand{\tabcolsep}{7pt}
\renewcommand{\arraystretch}{1.3}
\begin{tabular}{|ccccc|}
\hline
partial ring $R$ & prime $p$ & $\Idpt(R)$ & $\pSpec(R)$ & result\\ \hline
\hline
$\M_3(\F_p)_\sym$ &$p = 2,3$ & colorable & nonempty & Theorem~\ref{thm:finite sym coloring}\\
$\M_3(\Z)_\sym$ && colorable & nonempty & Corollary~\ref{cor:integer sym coloring} \\
$\M_3(\F_p)_\sym$ & $p \geq 5$ & uncolorable & empty & Theorems~\ref{thm:Bub}, \ref{thm:F5 sym coloring} \\
$\M_3(\Z[1/30])_\sym$ && uncolorable & empty & Theorem~\ref{thm:Bub} \\
$\M_3(\Q)_\sym$ && uncolorable & empty & Theorem~\ref{thm:Bub} \\
$\M_3(\Z)$ && uncolorable & empty & Theorem~\ref{thm:counting argument} \\ \hline
\end{tabular}
\endgroup
\end{table}

Our results may be of particular interest in the study of models of quantum mechanics defined over
fields other than the real or complex numbers. 
Quantum physics over $p$-adic fields has been a subject of interest for some time; for versions of 
$p$-adic quantum theory in which the amplitudes of wavefunctions are $p$-adic (as surveyed,
for instance, in~\cite[Section~9]{DKKV:padic}), the ``matrix entries'' of observable operators will 
also have $p$-adic values. 
Quantum physics over finite fields has also become a topic of recent interest, including investigations
involving modal quantum theory~\cite{SchumacherWestmoreland}, quantum 
computing~\cite{HOST:computing, HOST:theories}, and even quantum field 
theory~\cite{Schnetz, BejleriMarcolli}.
Our results show that some form of contextuality persists in either these settings, where observables
on finite-dimensional systems form matrices with entries over exotic commutative rings.

We wish to thank Alexandru Chirvasitu for several discussions and suggestions throughout the
writing of this paper, as well as Benno van den Berg and Chris Heunen for useful comments on a draft 
of the paper. We also thank the referee for several helpful suggestions.

\section{Partial algebraic structures and Kochen-Specker colorings}
\label{sec:partial}

We will largely follow the basic definitions of~\cite{KochenSpecker}, as adapted to the ring-theoretic
setting in~\cite{Reyes}. We follow the convention that every ring is associative and contains a multiplicative 
identity, and every ring homomorphisms preserves multiplicative identity elements.

The physical intuition for the terminology below is that a partial $k$-algebra $A$ consists of 
``observables of a quantum system,'' with $x,y \in A$ commeasurable if and only if the values of 
$x$ and $y$ can be simultaneously measured with arbitrarily high precision. 

\begin{definitions}
Let $k$ be a commutative ring.  A \emph{partial algebra} $A$ over $k$ is a set equipped with:
\begin{itemize}
\item a reflexive and symmetric binary relation $\comm \subseteq A \times A$, called \emph{commeasurability}, 
\item  ``partial'' addition and multiplication operations $+$ and $\cdot$ that are functions $\comm \to A$,
\item a scalar multiplication operation $k \times A \to A$, and
\item zero and unity elements $0,1 \in A$,
\end{itemize}
satisfying the following axioms:
\begin{enumerate}
\item $0$ and $1$ are commeasurable with all elements of $A$,
\item the partial binary operations preserve commeasurability,
\item for every pairwise commeasurable subset $S \subseteq A$, there exists a pairwise commeasurable
subset $T \subseteq A$ containing $S$ such that the restriction of the partial operations of $A$ make
$T$ into a (unital, associative) commutative $k$-algebra.
\end{enumerate}
If $A$ and $B$ are partial $k$-algebras, then a function $f \colon A \to B$ is called a \emph{morphism of partial
$k$-algebras} if $f(0) = 0$, $f(1) = 1$, $f(\lambda x) = \lambda f(x)$ for all $\lambda \in k$ and $x \in A$, 
and whenever $x,y \in A$ are such that $x \comm y$, it follows
that
\begin{itemize}
\item $f(x) \comm f(y)$ in $B$,
\item $f(x+y) = f(x) + f(y)$, and
\item $f(xy) = f(x)f(y)$.
\end{itemize}
A partial algebra over the ring $k = \Z$ is called a \emph{partial ring.}
A morphism of partial $\Z$-algebras is also called a \emph{morphism of partial rings}. The category of partial
rings with morphisms of partial rings is denoted $\pRing$.
\end{definitions}

We will use the terms \emph{total $k$-algebra} and \emph{total ring} to distinguish the usual notions
of $k$-algebra and ring from their partial counterparts. Every total $k$-algebra $R$ carries a natural
partial $k$-algebra structure, with the commeasurability relation given by $x \odot y$ if and only if $xy = yx$,
and with the restricted operations from the $k$-algebra structure of $R$. In this way we obtain a functor
 $\Ring \to \pRing$, which allows us to view the category of rings as a subcategory of that of partial
rings. Even though this partial algebra structure on a ring is not necessarily unique, we always 
view rings and algebras as partial rings and partial algebras with this canonical structure, trusting that
this will not lead to serious confusion.


We define a \emph{partial $k$-subalgebra} of a partial $k$-algebra $R$ to be a subset 
$S \subseteq R$ that is a partial $k$-algebra under the restricted commeasurability relation 
and partial operations inherited from $R$, or equivalently, such that such that $S$ is closed
under $k$-scalar multiplication as well as sums and products of commeasurable elements.
If all elements of $S$ are pairwise commeasurable, we say that $S$ is a \emph{commeasurable}
subalgebra; it is clear that $S$ becomes a total subalgebra under the induced operations.
In case $k = \Z$, we use the term \emph{partial subring} for a partial $\Z$-subalgebra.
(We note without further discussion that there is a subtlety in this terminology: it is possible 
to have a partial ring $S$, whose underlying set is a subset of a second partial ring $R$, such that the 
commeasurability relation on $S$ is strictly coarser than that of $R$. Then the inclusion 
function $i \colon S \to R$ is a morphism of partial rings, but $S$ is not a partial subring 
of $R$ in the sense above.) 

In the algebraic formulation of quantum mechanics, one views commutative algebras as corresponding
to classical systems. Then commeasurable subalgebras of a partial algebra can be seen as
``classical contexts'' in which a measurement may be performed on the corresponding quantum 
system~\cite[p.~48]{HLSW:Gelfand}. 
For more detail on this perspective, we refer readers to the recent survey~\cite{Heunen:classicalfaces}.

The \emph{(Zariski) spectrum} of a commutative ring $R$, denoted $\Spec(R)$,  is the set of prime 
ideals of $R$. As discussed in Section~\ref{sec:spectrum} below, the spectrum is a spatial 
invariant of a commutative ring; but for the time being, we will view it merely as a set. 
We recall the extension of the spectrum to an invariant of partial rings as in~\cite{Reyes}.

\begin{definition}
A subset $\p$ of a partial ring $R$ is a \emph{prime partial ideal} if, for every commeasurable
subring $C \subseteq R$, the intersection $C \cap \p$ is a prime ideal of $C$;  this is equivalent
to the conditions that $1 \notin \p$ and if $a,b \in R$ are commeasurable with $ab \in \p$, then
either $a \in \p$ or $b \in P$. The set of all prime partial ideals of $R$ is denoted $\pSpec(R)$.
\end{definition}

Given a morphism of partial rings $f \colon R \to S$ and $\p \in \pSpec(S)$, one may verify that
$f^{-1}(\p) \in \pSpec(R)$; see~\cite[Lemma~2.10]{Reyes}. Using this assignment on
morphisms, way we consider $\pSpec \colon \pRing\op \to \Set$ as a functor to the category of
sets. In particular, if we consider the category of rings as a subcategory of $\pRing$ via the 
functor $\Ring \to \pRing$ mentioned above, the partial spectrum restricts to the usual prime spectrum 
functor $\Spec \colon \Ringop \to \Set$. 

The remark and lemma below illustrate that the partial spectrum of a partial ring may be viewed as 
an invariant to detect obstructions of morphisms in $\pRing$ to total commutative rings, reminiscent
of Kochen and Specker's treatment of hidden variable theories.

\begin{remark}\label{rem:empty and nonempty}
It is well-known that every nonzero commutative ring has a (maximal, hence) prime ideal and consequently
has a nonempty spectrum; for instance, see~\cite[Theorem~1.1]{Matsumura}. 
On the other hand, the results of~\cite{Reyes} (to be generalized in Corollary~\ref{cor:any ring}) 
show that sufficiently large matrix rings have empty partial spectrum.
A useful technique to show that a particular partial ring $R$ has empty partial spectrum is to produce a
ring morphism of partial rings $R_0 \to R$ such that $\pSpec(R_0) = \varnothing$. For then functoriality
of the partial spectrum yields a function $\pSpec(R) \to \pSpec(R_0) = \varnothing$, and because the
only set with a function to the empty set is the empty set itself, we deduce $\pSpec(R) = \varnothing$.
\end{remark}

\begin{lemma}\label{lem:empty spec obstructions}
Given a partial ring $R$, if $\pSpec(R) = \varnothing$ then:
\begin{enumerate}
\item There is no morphism of partial rings $R \to C$ for any nonzero (total) commutative ring $C$;
\item The colimit in $\cRing$ of the diagram of commutative subrings of $R$ is zero.
\end{enumerate}
\end{lemma}

\begin{proof}
Let $f \colon R \to C$ be a morphism as in~(1) where $C \neq 0$. There exists a prime ideal 
$\p \in \Spec(C)$ as in Remark~\ref{rem:empty and nonempty}. Thus $f^{-1}(\p) \in \pSpec(R)$,
contradicting the assumption that $\pSpec(R) = \varnothing$. 

For~(2), let $L = \varinjlim C$ be the colimit in $\cRing$ of all commeasurable subrings 
$C \subseteq R$, equipped with canonical morphisms $f_C \colon C \to L$. 
Each $x \in R$ is contained in a commeasurable subring $C \subseteq R$, and the construction of
the colimit is such that the value $f(x) = f_C(x) \in L$ is independent of the choice of $C$. In this way
we obtain a well-defined function $f \colon R \to L$ that is readily verified to be a morphism of partial
rings (since $f$ restricts on each commeasurable subring to a ring homomorphism). It now follows
from~(1) that if $\pSpec(R) = \varnothing$, then $L = 0$.
\end{proof}

Thanks to the above, an important special case for us is the integer matrix ring $\M_n(\Z)$. This ring 
plays a \emph{universal} role in no-go theorems, due to the fact that each ring $R$ admits a unique 
ring homomorphism $\Z \to R$, which induces a ring homomorphism $\M_n(\Z) \to \M_n(R)$ by
when applied to each matrix entry. 
Thus a Kochen-Specker type of theorem proved for $\M_n(\Z)$ typically extends to $\M_n(R)$ for 
\emph{any} ring $R$. 

Kochen and Specker also considered ``partial logical structures'' in the following way. (We follow
the terse, but efficient, alternative definition given in~\cite{BergHeunen:colimit}.) 
Recall that a \emph{Boolean algebra} $(B,\vee,\wedge,\neg,0,1)$ is a structure such that
$(B,\vee,\wedge,0,1)$ is a distributive lattice with bottom element~$0$ and top element~$1$,
and a unary orthocomplement operation $\neg \colon B \to B$ (i.e., $\neg$ is an order-reversing
involution that maps each element to a lattice-theoretic complement).
The category $\Bool$ of Boolean algebras has as its morphisms the lattice homomorphisms
that preserve top and bottom elements along with the orthocomplement operation.
We refer readers to~\cite{Halmos, GivantHalmos} for the basic theory of Boolean algebras.

\begin{definition}
A \emph{partial Boolean algebra} is a set $B$ equipped with:
\begin{itemize}
\item a reflexive and symmetric commeasurability relation $\odot \subseteq B \times B$,
\item a unary operation of negation $\neg \colon B \to B$,
\item partially defined binary operations of meet and join $\wedge, \vee \colon \odot \to B$,
\item elements $0,1 \in B$,
\end{itemize}
such that every set $S \subseteq B$ of pairwise commeasurable elements is contained in a
pairwise commeasurable set $T \subseteq B$ containing $0$ and $1$ for which the restriction
of the operations makes $T$ into a Boolean algebra.
\end{definition}

As in the case of partial rings, we say that a subset of a partial Boolean algebra $B$ is a 
\emph{partial Boolean subalgebra} if it forms a partial Boolean algebra under the restricted
commeasurability relation and partial operations from $B$. We also use the terms
\emph{total Boolean algebra} and \emph{commeasurable Boolean subalgebra} to refer to
the obvious Boolean analogues of the corresponding ring-theoretic notions.

\begin{remark}
There is a classical correspondence~\cite[\S 2]{Halmos} between Boolean algebras and 
\emph{Boolean rings,} which are rings in which every element is idempotent. 
Define a \emph{partial Boolean ring} to be a partial ring in which every element is idempotent,
or equivalently, in which every commeasurable subring is Boolean; these form a full subcategory
of $\pRing$ which we denote by $\pBRing$.
The manner of defining Boolean algebra operations from a Boolean ring and of defining
Boolean ring operations from a Boolean algebra both extend directly by restricting to
commeasurable subalgebras or subrings. This yields an equivalence (even isomorphism!)\ of
categories $\pBool \cong \pBRing$.
\end{remark}

We say that two elements $p$ and $q$ of a partial Boolean algebra $B$ are \emph{orthogonal} if
they are commeasurable and $p \wedge q = 0$. Similarly, we define a partial ordering on $B$ by
declaring $p \leq q$ if $p$ and $q$ are commeasurable and $p \vee q = q$.

\begin{definition}
\label{def:KS coloring}
Let $B$ be a partial Boolean algebra with a subset $\setS \subseteq B$. A black-and-white coloring of 
$\setS$ is called a \emph{Kochen-Specker coloring} if, for every list of pairwise orthogonal elements 
$p_1,\dots,p_n \in B$, 
\begin{enumerate}
\item there is at most one index $i$ such that $p_i$ is colored white, and
\item if furthermore $p_1 \vee \cdots \vee p_n = 1$, then there is exactly one index $i$ such that
$p_i$ is colored white.
\end{enumerate}
\end{definition}

While the definition above is suited to an arbitrary subset of a partial Boolean algebra $B$, the algebraic
theory of such colorings is best behaved in the case where one takes $\setS = B$ to be the entire
algebra, as illustrated in the results outlined in the remainder of this section. 
On the other hand, to prove that the (possibly infinite) partial Boolean algebra $B$ has no 
Kochen-Specker colorings, it clearly suffices to exhibit a smaller (finite) subset $\setS$ of $B$ that has
no such coloring.

Notice immediately that for any Kochen-Specker coloring of a nontrivial ($0 \neq 1$) partial Boolean 
algebra, $0$ is black and $1$ is white by applying the condition above to the orthogonal decomposition 
$1 = 1 \vee 0 \vee 0$. On the other hand, the trivial Boolean algebra has no Kochen-Specker
coloring, as evidenced by condition~(2) applied to $0 = 0 \vee 1 = 1$.

\begin{remark}
\label{rem:vector coloring}
In the physics literature, Kochen-Specker colorings are usually considered on sets of vectors in
real or complex Hilbert spaces $H$. Such colorings can be considered as colorings subsets of the
orthomodular lattice of orthogonal projections on $H$ (viewed as a partial Boolean algebra as
in~\cite[Lemma~3.3]{HeunenReyes:active}, for instance) by identifying a vector $v$ with the rank-1 
orthogonal projection of $H$ onto the line spanned by $v$.
\end{remark} 

Next we aim to show that Kochen-Specker colorings on a piecewise Boolean $B$ algebra are intimately
related to the appropriate notion of prime ideals and ultrafilters of $B$. The suitable generalizations
of these objects are as follows.

\begin{definition}
\label{def:Boolean prime partial}
A subset $I$ of a partial Boolean algebra $B$ is called a \emph{partial ideal} if it satisfies the following 
conditions for commeasurable elements $p,q \in B$:
\begin{enumerate}
\item[(i)] $0 \in I$;
\item[(ii)] If $q \in I$ and $p \leq q$, then $p \in I$;
\item[(iii)] If $p, q \in I$, then $p \vee q \in I$.
\end{enumerate}
A partial ideal $I$ of $B$ is called a \emph{prime} partial ideal if it additionally satisfies the following 
condition for all commeasurable $p,q \in B$:
\begin{enumerate}
\item[(iv)] $1 \notin I$, and if $p \wedge q \in I$, then $p \in B$ or $q \in B$ (or equivalently,
for all $p \in B$, either $p$ or $\neg p$ is in $B$ but not both).
\end{enumerate}
The set of prime partial ideals of $B$ will be denoted $\pSpec(B)$.
\end{definition}

\begin{definition}
\label{def:partial ultrafilter}
A subset $F$ of a partial Boolean algebra $B$ is called a \emph{partial filter} if it satisfies the following 
conditions for commeasurable elements $p,q \in F$:
\begin{enumerate}
\item[(i)] $1 \in F$;
\item[(ii)] $p \in F$ and $p \leq q$ imply $q \in F$;
\item[(iii)] $p, q \in F$ implies $p \wedge q \in F$.
\end{enumerate}
A partial filter $F$ is called a \emph{partial ultrafilter} if it additionally satisfies the following condition 
for all commeasurable $p ,q \in B$:
\begin{enumerate}
\item[(iv)] $0 \notin B$ and if $p \vee q \in F$, then $p \in B$ or $q \in B$ (or equivalently, for all 
$p \in B$, either $p$ or $\neg p$ is in $F$ but not both).
\end{enumerate}
\end{definition}

The definitions above coincide with the usual definitions of  prime ideals and ultrafilters in case the 
partial Boolean algebra is in fact a total Boolean algebra. Furthermore, it is clear that 
a subset $X$ of a partial Boolean algebra $B$ is a partial ideal (respectively, prime partial ideal, partial 
filter, or partial ultrafilter) if and only if $X \cap C$ is an ideal (respectively, prime ideal, filter, or 
ultrafilter) of $C$ for  every commeasurable Boolean subalgebra $C$ of $B$. 
As in the classical case of total Boolean algebras, one may readily verify that a subset
$I$ of a partial Boolean algebra $B$ is a partial ideal if and only if $\neg I = \{\neg x \mid x \in I\}$ 
is a partial filter of $B$, and that $I$ is a prime partial ideal if and only if $\neg I$ is an ultrafilter,
if and only if $B \setminus I$ is an ultrafilter (equal to $\neg I$).

We also note that if $B$ is a partial Boolean algebra and $I$ is a subset of $B$, then $I$ is a prime
partial ideal of $B$ considered as a partial Boolean algebra if and only if $I$ is a prime partial ideal
of $B$ when considered as a partial Boolean ring. (This is perhaps most easily verified considering the
intersection $I \cap C$ for all commeasurable Boolean subalgebras $C \subseteq B$, and recalling 
that the two notions coincide in the classical case of total Boolean algebras and rings.)
Thus there is no danger in our use of the notation
$\pSpec(B)$ for the spectrum of prime partial ideals in both senses, as these two spectra in fact 
coincide. This assignment forms a functor $\pSpec \colon \pBool\op \to \Set$, which acts on a
morphism $f \colon A \to B$ by sending $\p \in \pSpec(B)$ to $\Spec(f)(\p) = f^{-1}(\p) \in
\pSpec(A)$.

\begin{proposition}
\label{prop:ultrafilter}
Let $B$ be a partial Boolean algebra, and fix a black-and-white coloring of $B$. The coloring is a 
Kochen-Specker coloring if and only if 
the set of white elements forms a partial ultrafilter of $B$, if and only if the set of black elements
forms a prime partial ideal of $B$.
\end{proposition}

\begin{proof}
First suppose that the coloring is Kochen-Specker; we verify the four conditions of 
Definition~\ref{def:partial ultrafilter} for the set of white elements.  Condition~(i) follows by 
applying the Kochen-Specker condition to the orthogonal decomposition $1 = 1 \vee 0 \vee 0$. 
Condition~(iv) follows easily by applying the Kochen-Specker condition to the orthogonal decomposition 
$1 = p \vee (\neg p)$. For condition~(ii), suppose that $p \leq q$ in $B$ with $p$ white. In the
orthogonal decomposition $1 = p \vee (q \wedge \neg p) \vee \neg q$, because $p$ is white the
third term $\neg q$ must be black, so that it follows from~(iv) that $q$ is white.
For~(iii), suppose that $p,q \in B$ are commeasurable and white. In the decomposition
\[
1 = (p \wedge q) \vee (p \wedge \neg q) \vee (q \wedge \neg p) \vee \neg(p \vee q),
\]
exactly one of the four joined terms on the right is white. If any of the second, third, or fourth
terms is white, then its complement is black by~(iv); as each of these elements $x$ has either
$p \leq x$ or $q \leq x$, we would deduce from condition~(ii) the contradiction that either 
$p$ or $q$ is black. Thus we must have $p \wedge q$ white as desired. 

Conversely, suppose that the set of white elements of the coloring satisfies conditions 
(i)--(iv) of Definition~\ref{def:partial ultrafilter}. To verify the Kochen-Specker condition, suppose 
\begin{equation}\label{eq:decomp}
1 = p_1 \vee \cdots \vee p_n
\end{equation}
for pairwise orthogonal elements $p_i \in B$;
we prove inductively that exactly one of the $p_i$ is white. The trivial case $n = 1$ follows from~(i). 
In case $n = 2$, the fact that exactly one of $p_1$ or $p_2 = \neg p_1$ is white follows from~(iv). 
Proceeding inductively, suppose the Kochen-Specker condition holds for all orthogonal decompositions 
of the unit into $n-1 \geq 2$ elements. We may rewrite~\eqref{eq:decomp} as 
\[
1 = (p_1 \vee p_2) \vee p_3 \vee \cdots \vee p_n
\]
and deduce by the inductive hypothesis that exactly one of $q = (p_1 \vee p_2), p_3, \dots, p_n$ is
white. If one of $p_3, \dots, p_n$ is white then $q$ is black. Applying~(ii) to 
$p_1, p_2 \leq q$ we obtain that $p_1$ and $p_2$ are both black, as desired. Now in case 
$p_3, \dots, p_n$ are black and $q$ is white, we only need to verify that exactly one of $p_1$ or $p_2$
is white. In the orthogonal decomposition $1 = q \vee \neg q = p_1 \vee p_2 \vee (\neg q)$, 
condition~(iv) implies that $\neg q$ is black. If $p_1$ and $p_2$ are both black, then 
$\neg (p_1 \vee p_2) = (\neg p_1) \wedge (\neg p_2)$ is a join of white elements and therefore
is white by~(iii), implying the contradiction that $p_1 \vee p_2 = q$ is black. Thus at least one of
$p_1$ or $p_2$ is white. Because $p_1 \wedge p_2 = 0 = \neg 1$ is black, condition~(iii) now shows
that only one of $p_1$ or $p_2$ can be white, as desired.

The set of white elements of the coloring is an ultrafilter if and only if its complement, the set of black 
elements, is a prime partial ideal. This completes the proof.
\end{proof}

Let $\KS(B)$ denote the set of Kochen-Specker colorings of a partial Boolean algebra $B$.
Given a morphism $f \colon B_1 \to B_2$ in $\pBool$ and a Kochen-Specker coloring of $B_2$, one
may readily verify using Proposition~\ref{prop:ultrafilter} that the coloring of $B_1$ given by declaring
$b \in B_1$ white if and only if $f(b) \in B_2$ is white yields a Kochen-Specker coloring of $B_1$.
Thus we obtain a functor $\KS \colon \pBool\op \to \Set$.

In the following, we let $\two = \{0,1\}$ denote the two-element Boolean algebra, which is the initial
object of both the category of Boolean algebras and $\pBool$.

\begin{theorem}
Let $B$ be a partial Boolean algebra.  Then the following three sets
are in bijection:
\begin{enumerate}
\item The set $\pBool(B,\two)$ of morphisms of partial Boolean algebras $B \to \two$;
\item The set $\KS(B)$ of Kochen-Specker colorings of $B$;
\item The set $\pSpec(B)$ of prime partial ideals of $B$.
\end{enumerate}
These bijections are natural in $B$ and thus form natural isomorphisms 
$\pBool(-,\two) \cong \KS \cong \pSpec$ as functors $\pBool\op \to \Set$.
\end{theorem}

\begin{proof}
We will define functions
\begin{equation}\label{eq:bijections}
\xymatrix{
& \pSpec(B) \ar[dr] & \\
\pBool(B,\two) \ar[ur] & & \KS(B) \ar[ll]
}
\end{equation}
as follows. Given $\phi \in \pBool(B,\two)$, set $\p_\phi = \phi^{-1}(0) \subseteq B$. Because
$\phi$ can equivalently be viewed as a morphism of partial Boolean rings $B \to \two$, we find 
that $\p_\phi = \pSpec(\phi)(0) \in \pSpec(B)$.

Next, given $\p \in \pSpec(B)$, it follows from Proposition~\ref{prop:ultrafilter} that the coloring of $B$ 
assigning black to all elements of $\p$ and white to all elements of $B \setminus \p$ is a 
Kochen-Specker coloring. This yields our function $\pSpec(B) \to \KS(B)$.

Finally, given a Kochen-Specker coloring of $B$, define a function $\phi \colon B \to \two$
by $\phi(b) = 0$ if $b \in B$ is colored black and $\phi(b) = 1$ if $b$ is colored white. Then
$\phi^{-1}(0)$ is a prime partial ideal of $B$ by Proposition~\ref{prop:ultrafilter}. Now the restriction of
$\phi$ to any commeasurable subalgebra $C$ of $B$ is such that $\phi|_C^{-1}(0) 
= \phi^{-1}(0) \cap C$ is a prime ideal, and it is well-known~\cite[Lemma~22.1]{GivantHalmos} 
that this implies that $\phi|_C$ is a homomorphism of Boolean algebras. So $\phi$ restricts to a 
Boolean algebra homomorphism on all commeasurable subalgebras, from which we conclude that 
it is a morphism in $\pBool(B,\two)$.

The composite of the three functions in the cycle~\eqref{eq:bijections} beginning at any of the three 
sets yields is readily seen to be the identity. Thus each of the functions is bijective.
Finally, it is straightforward to see from the construction of these bijections that they are natural in $B$,
yielding natural isomorphisms between the three functors $\pBool\op \to \Set$ as claimed.
\end{proof}

Given a partial ring $R$, let $\Idpt(R) = \{e \in R \mid e = e^2\}$ denote the set of idempotents 
elements of $R$. Given commeasurable elements $e,f \in \Idpt(R)$, it is straightforward to verify that
$e \vee f = e + f - ef$ and $e \wedge f = ef$ are both idempotents. Clearly $0,1 \in \Idpt(R)$ as well. 
It is straightforward to verify that the above operations endow $\Idpt(R)$ with the structure of a partial 
Boolean algebra. 
Furthermore, any morphism of partial rings $f \colon R \to S$ restricts to a morphism of partial
Boolean algebras $\Idpt(R) \to \Idpt(S)$. In this way we may view this assignment as a functor
\[
\Idpt \colon \pRing \to \pBool.
\]
In particular, if we again consider the category of rings to be a subcategory of $\pRing$, then this 
restricts to a functor $\Idpt \colon \Ring \to \pBool$.

It is straightforward to verify that with the above definitions in place, if $\p \in \pSpec(R)$ for a partial
ring $R$, then $\p \cap \Idpt(R)$ is a prime partial ideal of the partial Boolean algebra of idempotents.
In this way one obtains a natural transformation of functors $\pRing\op \to \Set$
\begin{equation}\label{eq:spectrum transformation}
\pSpec \to \pSpec \circ \Idpt \cong \KS \circ \Idpt.
\end{equation}
This allows us to deduce information about the partial spectrum of a ring from the (un)colorability
of its idempotents.

\begin{corollary}\label{cor:coloring to spectrum}
If $R$ is a partial ring such that the partial Boolean algebra $\Idpt(R)$ has no Kochen-Specker
colorings, then $\pSpec(R) = \varnothing$.
\end{corollary}

\begin{proof}
The natural transformation~\eqref{eq:spectrum transformation} provides a function
$\pSpec(R) \to \KS(\Idpt(R))$. Because the latter set is empty, so is the former. 
\end{proof}

A square matrix over a commutative ring is said to be \emph{symmetric} if it is equal
to its own transpose.
For a commutative ring $R$ and positive integer $n$, we let $\M_n(R)_\sym$ denote 
the subset of $\M_n(R)$ consisting of symmetric matrices, and we let 
$\Proj(R) = \Idpt(\M_n(R)_\sym)$ denote the set of symmetric idempotents, which
we call \emph{projections}.
It is clear that $\M_n(R)_\sym$ is a partial $R$-subalgebra of $\M_n(R)$, and that
$\Proj(\M_n(R))$ is a partial Boolean subalgebra of $\Idpt(\M_n(R))$. 
Together, we obtain a diagram of sets
\[
\xymatrix{
\pSpec(\M_n(R)_\sym) \ar[d] & \pSpec(\M_n(R)) \ar[l] \ar[d] \\
\pSpec(\Proj(\M_n(R)) & \pSpec(\Idpt(\M_n(R)) \ar[l]
}
\]
that is easily shown to commute.
Thus, to show that $\pSpec(\M_n(R)) = \varnothing$, it suffices to show that any one 
of the other three partial spectra is empty.

The next lemma shows that the nonexistence of either Kochen-Specker colorings or of prime partial
ideals extends from matrix rings of a fixed order to all matrix rings of larger order. Throughout
the following, for a ring $R$, we let $E_{ij} \in \M_n(R)$ denote the matrix unit whose $(i,j)$-entry
is~$1$ and whose other entries are~$0$.

\begin{lemma}\label{lem:obstruction passes on}
Let $R$ be a ring and let $m \geq 1$ be an integer.
\begin{enumerate}
\item If $\Idpt(\M_m(R))$ has no Kochen-Specker colorings, then also $\Idpt(\M_n(R))$ has no 
Kochen-Specker colorings for all integers $n \geq m$.
\item If $\pSpec(\M_m(R)) = \varnothing$, then also $\pSpec(\M_n(R)) = \varnothing$ for all integers 
$n \geq m$.
\end{enumerate}
Now assume furthermore that $R$ is commutative. 
\begin{enumerate}
\item[(3)] If $\Proj(\M_m(R))$ has no Kochen-Specker colorings, then also $\Proj(\M_n(R))$ has no 
Kochen-Specker colorings for all integers $n \geq m$.
\item[(4)] If $\pSpec(\M_m(R)_\sym) = \varnothing$, then also $\pSpec(\M_n(R)_\sym) = \varnothing$ 
for all integers $n \geq m$.
\end{enumerate}
\end{lemma}

\begin{proof}
To prove~(1), it suffices by induction assume that $\Idpt(\M_n(R))$ has no Kochen-Specker coloring and
deduce that $\Idpt(M)$ has no such coloring for $M = \M_{n+1}(R)$. Suppose toward a contradiction 
that there $\Idpt(M)$ does have a Kochen-Specker coloring. 
Consider the diagonal idempotents $E_{ii}$ for $i = 1, \dots, n, n+1$.
It follows that exactly one of the $E_{ii}$ is white; assume without loss of generality that this is $E_{11}$. 

For the idempotent $E = E_{11} + \cdots + E_{nn} = 1 - E_{n+1,n+1} \in M$, because $E_{11} \leq E$ in
$\Idpt(M)$ we have that $E$ is white according to Proposition~\ref{prop:ultrafilter}. 
The corner ring $EME$ has multiplicative identity $E$, so that the restriction of the coloring to 
$\Idpt(EME)$ satisfies condition~(i) of Proposition~\ref{prop:ultrafilter}. But conditions~(ii)--(iv) of the same 
lemma are easily seen to pass to $\Idpt(EME)$, from which it follows that this restriction is a
Kochen-Specker coloring. 
But it is clear from the choice of $E$ that $EME \cong \M_n(R)$.  Thus we obtain the contradiction that 
$\Idpt(\M_n(R)) \cong \Idpt(EME)$ has a Kochen-Specker coloring.

The proof of~(2) also proceeds inductively, showing that any $\p \in \pSpec(M)$ for $M = 
\M_{n+1}(R)$ induces a prime partial ideal of $\M_n(R)$. This is done similarly to the proof of part~(1) 
above, this time noting that without loss of generality we may assume $E_{11} \in \p$ with the rest of 
the $E_{ii} \notin \p$, so that for $E = E_{11} + \cdots + E_{nn}$ the restriction $\p \cap EME$ is a 
prime partial  ideal of $E M E \cong \M_n(R)$.

The proofs of~(3) and~(4) follow the same arguments as those given above with only minor
modifications: the idempotents $E_{ii}$ and $E$ above are symmetric, so that the
transpose restricts to an involution of the corner ring $EME$, resulting in isomorphisms
$EME_\sym \cong \M_n(R)_\sym$.
\end{proof}

\begin{remark}\label{rem:constructably uncolorable}
If $\setS \subseteq \Idpt(\M_n(R))$ is a subset that has no Kochen-Specker coloring,
then one may adapt the proof above to explicitly construct a new set $\setS^+ \subseteq 
\Idpt(\M_{n+1}(R))$ that also has no Kochen-Specker coloring, by taking the diagonal matrix units 
$E_{ii}$ for $i = 1, \dots, n,n+1$ along with isomorphic copies of $\setS$ in each of the partial 
Boolean algebras $\Idpt((1-E_{ii}) \M_{n+1}(R) (1-E_{ii})) \cong \Idpt(\M_n(R))$.
\end{remark}

\section{Colorability of idempotent matrices over various rings}
\label{sec:coloring}

In the following, for a field $F$, we consider the $F$-vector spaces $F^n$
to consist of column vectors.
Given vectors $u,v \in F^n$, we denote their usual ``dot product'' by
\[
u \cdot v = u^T v = \sum u_i v_i.
\]
This defines a bilinear form on $F^n$, but this may be a degenerate pairing depending
upon the ground field $F$.

\begin{lemma}\label{lem:denominator}
Let $F$ be a field, and let $v \in F^n \setminus \{0\}$. Then the following are equivalent:
\begin{enumerate}
\item There is a symmetric idempotent in $\M_n(F)$ with range $\Span(v)$;
\item The sum of squares $v^T v \in F$ is nonzero.
\end{enumerate}
In case $v^T v = \lambda \neq 0$, the symmetric idempotent with range
$\Span(v)$ is $P_v = \lambda^{-1} v v^T$. Finally, given $u,v \in F^n$ with
$u^T u \neq 0 \neq v^T v$, the projections $P_u$ and $P_v$ are orthogonal
if and only if $u \cdot v = 0$.
\end{lemma}

\begin{proof}
Assume~(1) holds, so that $P = P^2 = P^T \in \M_n(F)$ with $\range(P) = \Span(v)$. Because
$v \neq 0$, we have $P \neq 0$. Thus some entry of $P$ is nonzero, say the $(i,j)$-entry. 
Let $v_i$ and $v_j$ denote the $i$th and $j$th rows of $P$ respectively. Then the $(i,j)$-entry
of $P = P^2 = P^T P$ is equal to $v_i^T v_j \neq 0$. But as $v_i, v_j \in \range(P) = \Span(v)$, 
the product $v_i^T v_j$ is a scalar multiple of $v^T v$. Thus~(2) must hold.

Conversely, suppose~(2) holds, and set $\lambda = v^T v \neq 0$. Then $P = \lambda^{-1} v v^T$
satisfies $P = P^T$ and 
\[
P^2 = (\lambda^{-1} vv^T) (\lambda^{-1} vv^T) = \lambda^{-2}  v (v^T v) v^T = \lambda^{-1} v v^T = P.
\]
Given any $w \in F^n$, since $Pw = \lambda^{-1} vv^Tw = (\lambda^{-1} v \cdot w)v \in \Span(v)$
and $Pv = v$,  we see that $\range(P) = \Span(v)$. Thus~(1) holds.

Finally, suppose $u,v \in F^n$ are as in the last sentence of the lemma. If $u \cdot v = 0$ then we have
\[
P_u P_v = (u^Tu)^{-1} u u^T \cdot (v^Tv)^{-1} vv^T = (u^T u \cdot v^Tv)^{-1} u (u^T v) v^T = 0,
\]
and $P_v P_u = (P_u P_v)^T = 0$. Conversely, suppose that $P_u P_v = 0$. Then 
\[
u \cdot v = (P_u u) \cdot (P_v v) = (P_v^T P_u u) \cdot v = (P_v P_u u) \cdot v = 0
\]
as desired.
\end{proof}

The lemma above allows us to equate Kochen-Specker colorings of rank-1 symmetric projections
over a field $F$ with Kochen-Specker colorings of vectors $v$ satisfying $v^T v \neq 0$, in a
manner analogous to Remark~\ref{rem:vector coloring}.

Our first uncolorability result makes use of one of the few vector configurations in
the literature on Kochen-Specker colorings for which all vectors have integer entries. An account is given in~\cite{Bub:Schutte}
(also~\cite[Chapter~3]{Bub:interpreting}) of a proof of the Kochen-Specker Theorem due to Kurt 
Sch\"{u}tte, making use of a certain classical tautology that does not remain a tautology when 
interpreted in the partial Boolean algebra $\Proj(\M_3(\R))$. 
The orthogonal projections used to represent this logical proposition happen
to be projections onto lines spanned by vectors with integer entries. Though these vectors do not have
unit length, and their normalizations have irrational entries, we observe in the proof below that the
resulting projection matrices do in fact have rational entries. We recall below that for any rational
$q \in \Q$, the ring $\Z[q]$ denotes the subring of $\Q$ generated by the integers and $q$, and
consists of elements of the form $f(q)$ where $f$ is any polynomial with integer entries.

\begin{theorem}\label{thm:Bub}
The partial Boolean algebra $\Proj(\M_3(\Z[1/30])$ has no Kochen-Specker coloring.
Consequently, the partial ring $\M_3(\Z[1/30])_\sym$ has no prime partial ideals.
\end{theorem}

\begin{proof}
Consider the uncolorable set of vectors~\cite[Section~4]{Bub:Schutte} used in Sch\"{u}tte's proof
of the Kochen-Specker theorem. These are vectors in $\Z^3$, and each of the vectors $v$
is such that $\|v\|^2 = v^Tv$ divides $30$. Thus each of the corresponding orthogonal projections
$p_v = (v^Tv)^{-1} vv^T$ lies in $\Proj(\M_3(\Z[1/30])$ by Lemma~\ref{lem:denominator}. 
Thus the argument of Sch\"{u}tte and Bub in fact shows more generally that $\Proj(\M_3(\Z[1/30]))$ has no 
Kochen-Specker coloring.

Because $\Proj(\M_3(\Z[1/30])) = \Idpt(\M_3(\Z[1/30])_\sym)$, the claim about prime partial ideals 
follows from Corollary~\ref{cor:coloring to spectrum}.
\end{proof}

Let $F$ be a field and consider the canonical ring homomorphism $\Z \to F$. If the characteristic
of $F$ does not divide $30$, then this homomorphism factors uniquely as
$\Z \to \Z[1/30] \to F$. This induces a ring homomorphism $\M_3(\Z[1/30]) \to \M_3(F)$, along
with a morphism of partial Boolean algebras $\Proj(\M_3(\Z[1/30])) \to \Proj(\M_3(F))$. By
functoriality of Kochen-Specker colorings, the theorem above implies that $\Proj(\M_3(F))$ has 
no Kochen-Specker colorings and therefore that $\M_3(F)_\sym$ has no prime partial ideals.
In particular, these remarks apply to the field $F = \Q$ of rational numbers. 

\begin{remark}
In the literature addressing finite-precision loopholes to the Kochen-Specker theorem (as
in~\cite{Meyer, CliftonKent} and many further references discussed in~\cite{BarrettKent})
it is well-documented that the set $S = \Q^3 \cap S^2$ of vectors with rational coordinates on
the unit sphere has a Kochen-Specker coloring~\cite{GodsilZaks}.
The apparent conflict between this fact and the uncolorability of $\Proj(\M_3(\Q))$ may be resolved
as follows.
The mapping $\phi \colon S \to \Proj(\M_3(\Q))$ given by $\phi(v) = P_v = \|v\|^{-2} vv^T 
= vv^T$ preserves orthogonality by Lemma~\ref{lem:denominator} and has image contained in the
rank-1 rational projections. By the same lemma, every rank-1 projection in $\Proj(\M_3(\Q))$
is of the form $P = P_v = \|v\|^{-2} vv^T$ for any nonzero vector $v$ in the range of $P$. 
The image of $\phi$ forms a proper subset of the rank-1 rational projections, as one readily
verifies that for vectors such as $v = \begin{pmatrix} 1 & 1 & 1 \end{pmatrix}^T$ such that $v/\|v\|$
has irrational entries, the projection $P_v$ lies outside of the image of $\phi$. So the rational
unit vectors correspond to a Kochen-Specker colorable subset of the Kochen-Specker uncolorable
set of \emph{all} rank-1 rational projections.
\end{remark}

On the other hand, if $F$ has characteristic $p$ dividing $30$ (i.e., $p = 2,3,5$), then the
ring homomorphism $\Z \to F$ \emph{does not} factor through $\Z[1/30]$, so we cannot
make the same conclusion about Kochen-Specker colorings of projections or prime partial
ideals in $\M_3(F)_\sym$. 
In the following we use $\F_q$ to denote the finite field with $q$ elements. In case 
$F = \F_p$ we will show below that such Kochen-Specker colorings and prime partial ideals 
do in fact exist in case $p = 2,3$ but not in case $p = 5$.

\begin{theorem}\label{thm:finite sym coloring}
There exist Kochen-Specker colorings of $\Proj(\M_3(\F_p))$ and prime partial ideals 
of the partial rings $\M_3(\F_p)_\sym$ for $p = 2,3$. 
\end{theorem}

\begin{proof}
We establish the existence of Kochen-Specker colorings below. Then it will follow from
Lemma~\ref{lem:coloring to morphism} that there exists a morphism of partial $F$-algebras
$\phi \colon \M_3(\F_p)_\sym \to K$ for a field extension $K$ of $\F_p$, making $\phi^{-1}(0)$
a prime partial ideal of $\M_3(\F_p)_\sym$ for $p = 2,3$.

$p = 2$: There are four vectors $v \in \F_2^3$ satisfying $v^Tv \neq 0$, yielding four rank-1
projections in $\M_3(\F_2)$; three of these projections are the diagonal matrix units $E_{ii} = e_i e_i^T$
from the standard basis vectors $\{e_i \mid i = 1,2,3\}$, and the fourth is
\[
U = uu^T = \begin{pmatrix} 1 & 1 & 1 \\ 1 & 1 & 1 \\ 1 & 1 & 1 \end{pmatrix} 
\quad \mbox{for} \quad u = \begin{pmatrix} 1 \\ 1 \\ 1 \end{pmatrix}.
\]
Thus $\Proj(\M_3(\F_2))$ has two maximal commeasurable Boolean subalgebras: one is generated
by the $E_{ii}$ (and therefore isomorphic to the power set algebra on a three-element set), and
the other is given by $\{0, U, I-U, I\}$ (isomorphic to the power set of a two-element set). Now any
independent choice of a Kochen-Specker coloring on each of these maximal commeasurable subalgebras
(given by any homomorphism into $\two$) gives a Kochen-Specker coloring of $\Proj(\M_3(\F_2))$.

$p = 3$: In this case every projection in $\M_3(\F_3)$ is a sum of orthogonal rank-$1$ projections,
of which there are nine. In fact, there are only four orthogonal triples of rank-$1$ projections that
sum to the identity (the off-diagonal entries that are omitted below are zero):
\begingroup
\allowdisplaybreaks
\begin{align*}
I &= 
\begin{pmatrix}
1 & & \\
 & 0 & \\
 & & 0
\end{pmatrix}
+
\begin{pmatrix}
0 & & \\
 & 1 & \\
 & & 0
\end{pmatrix}
+
\begin{pmatrix}
0 & & \\
 & 0 & \\
 & & 1
\end{pmatrix} \\
&= 
\begin{pmatrix}
1 & & \\
 & 0 & \\
 & & 0
\end{pmatrix}
+
\begin{pmatrix}
0 & & \\
 & 2 & 2\\
 & 2 & 2
\end{pmatrix}
+
\begin{pmatrix}
0 & & \\
 & 2 & 1 \\
 & 1 & 2
\end{pmatrix} \\
&= 
\begin{pmatrix}
0 & & \\
 & 1 & \\
 & & 0
\end{pmatrix}
+
\begin{pmatrix}
2 & 0 & 2 \\
0 & 0 & 0 \\
2 & 0 & 2
\end{pmatrix}
+
\begin{pmatrix}
2 & 0 & 1 \\
0 & 0 & 0 \\
1 & 0 & 2
\end{pmatrix} \\
&= 
\begin{pmatrix}
0 & & \\
 & 0 & \\
 & & 1
\end{pmatrix}
+
\begin{pmatrix}
2 & 2 & \\
2 & 2 & \\
 &  & 0
\end{pmatrix}
+
\begin{pmatrix}
2 & 1 & \\
1 & 2 &  \\
 & & 0
\end{pmatrix}
\end{align*}
\endgroup
As each of the diagonal matrix units $E_{ii}$ is contained in $T = \{E_{11}, E_{22}, E_{33}\}$ and exactly
one other orthogonal triple, it is easy to verify that any Kochen-Specker coloring of the triple
$T$ can be extended to a Kochen-Specker coloring of the nine projections above and thereby 
to all of $\Proj(\M_3(\F_3))$.
\end{proof}

\begin{corollary}\label{cor:integer sym coloring}
There exist morphisms of partial rings $\M_3(\Z)_\sym \to \F_{p^6}$ for $p = 2,3$. Consequently,
$\pSpec(\M_3(\Z)_\sym) \neq \varnothing$.
\end{corollary}

\begin{proof}
For $p = 2,3$ let $\M_3(\Z)_\sym \to \M_3(\F_p)_\sym$ be the canonical homomorphism that acts
``modulo~$p$'' in each matrix entry. By Theorem~\ref{thm:finite sym coloring} and 
Lemma~\ref{lem:coloring to morphism} (see also Remark~\ref{rem:finite extension}), there is a 
morphism of partial $\F_p$-algebras $\M_3(\F_p) \to \F_{p^6}$, since $\F_{p^6}$ is up to isomorphism 
the unique degree six extension of $\F_p$. The composite of these morphisms yields the desired 
function, and the preimage of the zero ideal of $\F_{p^6}$ is a prime partial ideal of $\M_3(\Z)_\sym$.
\end{proof}

Next we will show that the projections in $\M_3(\F_5)$ do not have a Kochen-Specker coloring.
To this end, we note that each $v \in \F_5^3$ satisfying $v^T v \neq 0$ is a scalar multiple of a 
unique vector in the list below.
\begingroup
\allowdisplaybreaks
\addtolength{\jot}{1em}
\begin{align*}
v_1 &= \colvec[1]{0}{0} & v_2 &= \colvec[0]{1}{0} & v_3 &= \colvec[0]{0}{1} & v_4 &= \colvec[1]{1}{0} & v_5 &= \colvec[1]{0}{1} \\
v_6 &= \colvec[0]{1}{1} & v_7 &= \colvec[-1]{1}{0} & v_8 &= \colvec[-1]{0}{1} & v_9 &= \colvec[0]{-1}{1} & v_{10} &= \colvec[1]{1}{1} \\
v_{11} &= \colvec[-1]{1}{1} & v_{12} &= \colvec[1]{-1}{1} & v_{13} &= \colvec[1]{1}{-1} & v_{14} &= \colvec[2]{1}{1} & v_{15} &= \colvec[1]{2}{1} \\
v_{16} &= \colvec[1]{1}{2} & v_{17} &= \colvec[-2]{1}{1} & v_{18} &= \colvec[1]{-2}{1} & v_{19} &= \colvec[1]{1}{-2} & v_{20} &= \colvec[-2]{2}{1} \\
v_{21} &= \colvec[-2]{1}{2} & v_{22} &= \colvec[1]{-2}{2} & v_{23} &= \colvec[2]{-2}{1} & v_{24} &= \colvec[2]{1}{-2} & v_{25} &= \colvec[1]{2}{-2}
\end{align*}
\endgroup
Thus the $25$ rank-1 symmetric projections are exactly those of the form 
$P_i = P_{v_i}$ for the vectors $v_i$ above.

Suppose that $Q \in \M_3(\F_5)$ is an invertible matrix with $Q^{-1} = Q^T$. If 
$A \in \M_3(\F_5)_\sym$ then $(QAQ^{-1})^T = (Q^{-1})^T A^T Q^T = QAQ^{-1}$, so that
conjugation by $Q$ restricts from an automorphism of the $\F_5$-algebra $\M_3(\F_5)$ to a 
partial $\F_5$-algebra automorphism of $\M_3(\F_5)^\sym$ and to an automorphism of the 
partial Boolean algebra $\Proj(\M_3(\F_5))$. This applies in particular if $Q$ is any
permutation matrix.

\begin{theorem}
\label{thm:F5 sym coloring}
There is no Kochen-Specker coloring of $\Proj(\M_3(\F_5))$. Thus $\M_3(\F_5)_\sym$ has
no prime partial ideals.
\end{theorem}

\begin{proof}
Assume for contradiction that there is a Kochen-Specker coloring of the rank-1 projections
in $\M_3(\F_5)$; this induces a coloring on the vectors $\{v_i \mid i = 1, \dots, 25\}$ above.

As remarked above, conjugation by any $3 \times 3$-permutation matrix restricts to an
automorphism of the partial Boolean algebra $\Idpt(\M_3(\F_5))$. This automorphism
clearly preserves the rank of a projection. So it restricts to a bijection on the set of rank-1
projections $\{P_i \mid i = 1,\dots,25\}$ defined by the vectors above. These automorphisms
permute the diagonal matrix units $E_{ii} = P_i$ for $i = 1,2,3$. Thus after conjugating by
an appropriate permutation matrix, we may assume that the coloring of the vectors is such
that $v_1, v_2$ are black and $v_3$ is white. By orthogonality of the triple $\{v_3, v_4, v_7\}$,
we must have $v_4$ and $v_7$ colored black.

Similarly, conjugation by the symmetric matrix 
$Q = \left(\begin{smallmatrix} 1 & 0 & 0 \\ 0 & -1 & 0 \\ 0 & 0 & 1 \end{smallmatrix}\right)$ 
restricts to an automorphism of $\Proj(\M_3(\F_5))$ that preserves rank. This automorphism 
fixes $P_1$ (as $Qv_1= v_1$) and permutes $P_6$ with $P_9$ (as $Qv_6 = v_9$ and 
$Qv_9 = v_6$). Thus, after conjugating by $Q$ if necessary, we assume without loss of generality that 
$v_6$ is colored black and $v_9$ is white. From the following orthogonality relations, we deduce 
the colorings below:
\begingroup
\allowdisplaybreaks
\begin{align*}
\{v_9, v_{10}, v_{17}\} \mbox{ orthogonal} &\implies v_{10} \mbox{ black}, \\
\{v_9, v_{11}, v_{14}\} \mbox{ orthogonal} &\implies v_{11} \mbox{ black}, \\
\{v_4, v_{11}, v_{20}\} \mbox{ orthogonal} &\implies v_{20} \mbox{ white}, \\
\{v_7, v_{10}, v_{19}\} \mbox{ orthogonal} &\implies v_{19} \mbox{ white}, \\
\{v_{18}, v_{20}, v_{25}\} \mbox{ orthogonal} &\implies v_{25} \mbox{ black}, \\
\{v_{19}, v_{21}, v_{22}\} \mbox{ orthogonal} &\implies v_{21} \mbox{ black}, \\
\{v_5, v_{11}, v_{21}\} \mbox{ orthogonal} &\implies v_5 \mbox{ white}, \\
\{v_6, v_{13}, v_{25}\} \mbox{ orthogonal} &\implies v_{13} \mbox{ white}. \\
\end{align*}
\endgroup
But now orthogonality of the triple $\{v_5, v_{13}, v_{24}\}$ contradicts the coloring of the vectors
$v_5$ and $v_{13}$ above, establishing uncolorability of the set of projections $\{P_i\}_{i=1}^{25}$
and consequently the uncolorability of $\Proj(\M_3(\F_5))$. 

Because $\Proj(\M_3(\F_5)) = \Idpt(\M_3(\F_5)_\sym)$, the second claim follows from the first by
Corollary~\ref{cor:coloring to spectrum}.
\end{proof}

Our results on Kochen-Specker colorings on symmetric idempotents over finite fields has the
following consequence for Kochen-Specker colorings of vectors in $\R^3$ whose coordinate entries
happen to be integers (such as those considered in~\cite{Bub:Schutte}).

\begin{corollary}
Suppose that $\{v_i\}$ is a set of vectors in $\R^3$ for which there is no Kochen-Specker coloring. If
all $v_i$ have integer coordinates, then the least common multiple of the integers 
$v_i \cdot v_i = \|v_i\|^2$ is divisible by $6$.
\end{corollary}

\begin{proof}
Let $N = \operatorname{lcm}\{\|v_i\|^2\}$ be the least common multiple described in the statement.
Given any prime $p$, entrywise application of the canonical map $\Z \twoheadrightarrow \F_p$, 
denoted $x \mapsto \overline{x}$, induces a linear mapping $\Z^3 \to \F_p^3$, which we similarly 
denote by $v \mapsto \overline{v}$. Suppose that $p \nmid N$; then each of the images
$\overline{v_i} \cdot \overline{v_i} = \overline{v_i \cdot v_i} \in \F_p$ are nonzero. Thus we obtain 
projections $Q_i = (\overline{v_i \cdot v_i})^{-1} \overline{v_i} \overline{v_i}^T$ in $\Proj(\M_3(\F_p))$ 
for all $i$, with $Q_i$ and $Q_j$ orthogonal whenever $v_i$ is orthogonal to $v_j$. So if the set
$\{v_i\}$ has no Kochen-Specker coloring, the same must be true of the set $\{Q_i\}$, making
$\Proj(\M_3(\F_p))$ uncolorable.

But now it follows from Theorem~\ref{thm:finite sym coloring} that $N$ must be divisible by both
$p = 2$ and $p = 3$, yielding the desired result.
\end{proof}

Related to the results presented above, we ask two related questions:

\begin{question}
(A) Can the conclusion of the corollary above  be strengthened to state that 
$30 = 2 \cdot 3 \cdot 5$ must divide the least common multiple of the $\|v_i\|^2$?

(B) Does $\Proj(\M_3(\Z[1/6]))$ have a Kochen-Specker coloring?
\end{question}

The method of proof used above does not extend to the prime $p = 5$ because of the uncolorability of
$\Proj(\M_3(\F_5))$, which leads to question (A). 
Note that every rank-1 projection in $\M_3(\Z[1/6])$ is of the form $\|v\|^{-2} v v^T$ for
some $v \in \Z^3$. Choosing $v$ to have the least common multiple of its entries equal to~$1$,
then we may conclude that $\|v\|^2 = v \cdot v$ divides a power of~$6$. So a negative
answer to question~(B) would imply a negative answer to question~(A).

\separate

We now turn our attention to partial Boolean algebras of (non-symmetric) idempotents
over various rings, beginning with finite fields. While Theorem~\ref{thm:counting argument}
will be generalized in Corollary~\ref{cor:any ring} below, the uncolorable set of idempotents 
to be constructed for the proof of Theorem~\ref{thm:integer KS} is actually motivated
by this preliminary result regarding finite fields.

The following theorem was communicated to us by Alexandru Chirvasitu, whom we thank for
kindly allowing us to include it here. 
Its proof uses some well-known methods of counting subspaces in vector spaces over finite 
fields via Gaussian binomial coefficients; see~\cite{Prasad}, for instance.
Let $q$ be a prime power, let $F = \F_q$ be the field of $q$ elements, and let $V$
be an $F$-vector space of dimension $n$. 
The number of $k$-dimensional subspaces of $V$ is given by the number of ordered
linearly independent lists of $k$ vectors in $V$ (each of which spans a $k$-dimensional
subspace) divided by the number of bases for a $k$-dimensional vector space:
\[
\binom{n}{k}_q = \frac{(q^n - 1)(q^n - q) \cdots (q^n - q^{k-1})}{(q^k - 1)(q^k - q) 
  \cdots (q^k - q^{k-1})}.
\]

The dual vector space $V^* = \Hom_F(V,F)$ is also an $n$-dimensional vector space.
Given a subspace $W \subseteq V$ of dimension $k$, there is a corresponding
subspace $W^\perp = \{f \in V* \mid f(W) = 0\}$ of dimension $n - k$ in $V^*$.
Note that $W_1 \subseteq W_2$ implies $W_2^\perp \subseteq W_1$ for subspaces
$W_i$ of $V$. Because $V \cong V^*$ as vector spaces, this provides a bijection
between the subspaces of $V$ having dimension $k$ and the subspaces of $V$
having dimension $n -k$, which reverses inclusion of subspaces.

For instance, when $V = F^3$ so that $n = 3$, the number of 1-dimensional
subspaces in $V$ is $\binom{3}{1}_q = (q^3-1)/(q-1) = q^2 + q + 1$ and the
number of 2-dimensional subspaces is also $\binom{3}{2}_q = q^2 + q + 1$.
Also, the number of 2-dimensional subspaces of $V$ that contain a given
1-dimensional subspace is equal to the number of 1-dimensional subspaces
contained in a given 2-dimensional subspace, and is therefore equal to
$\binom{2}{1}_q = (q^2 - 1)/(q-1) = q+1$.

\begin{theorem}\label{thm:counting argument}
Let $p$ be a prime such that $p \equiv 2 \pmod{3}$. Then $\Idpt(\M_3(\F_p))$ has no
Kochen-Specker coloring. 
\end{theorem}

\begin{proof}
We prove the contrapositive. 
Suppose that there exists a Kochen-Specker coloring of $\Idpt(\M_3(\F_p))$.
Consider the set $\setS$ of unordered orthogonal triples $\{E_1,E_2,E_3\}$ of rank-1 
idempotents in $\M_3(\F_p)$ such that $E_1 + E_2 + E_3 = I$.
Given a rank-1 idempotent $E$, let $\setS_E = \{ T \in \setS \mid E \in T\}$;
we claim that these sets have the same cardinality for all $E$.
Indeed, the general linear group $G = \GL_3(\F_p)$ acts transitively on the rank-1
idempotents in $\M_3(\F_p)$ via conjugation, and this induces an action on $\setS$ by
\[
{}^U T = \{U E_i U^{-1} \mid i = 1,2,3\} \quad \mbox{for } U \in G \mbox{ and }
T = \{E_1,E_2,E_3\} \in \setS.
\]
Now given rank-1 idempotents $E$ and $F$, if we fix $U \in G$ with $U E U^{-1} = F$,
the action of $U$ on $\setS$ carries the elements of $\setS_E$ to the elements of
$\setS_F$. Thus these two sets are in bijection, establishing the claim. We let
$N = |\setS_E|$ denote the number of triples in $\setS$ that contain any given
rank-1 idempotent $E$, independent of $E$.

By the Kochen-Specker property of the coloring, each $T \in \setS$ contains exactly
one white idempotent and two black idempotents. 
Thus the number of white rank-1 idempotents is equal to $|\setS|/N$ and the number 
of black rank-1 idempotents is equal to $2(|\setS|/N)$. In particular, the
number of rank-1 idempotents is $3(|\setS|/N)$, a multiple of~$3$.

On the other hand, each rank-1 idempotent is uniquely determined by the choice of its 
range, which can be any line in the vector space $V = \F_p^3$, along with its kernel, which
can be any plane in $V$ not containing that line. The number of lines in $V$ is equal to 
$\binom{3}{1}_p = p^2 + p + 1$. The number of planes not containing a given line in 
$V$ is equal (by duality) to the number of lines not contained in a given plane, and is
therefore equal to $\binom{3}{1}_p - \binom{2}{1}_p = (p^2 + p + 1) - (p+1) = p^2$.
Thus we may alternatively calculate the number of rank-1  idempotents in $\M_3(\F_p)$ 
to be equal to $(p^2 + p + 1)p^2$. 

It follows that $3$ divides $(p^2 + p + 1)p^2$. This is only possible if 
$p \not\equiv 2 \pmod{3}$, completing the proof.
\end{proof}

At this point, we are able to deduce that \emph{for any prime $p \neq 3$, there is no 
Kochen-Specker coloring of $\Idpt(\M_3(\F_p))$}. Indeed, if $p \notin \{2,3,5\}$, then the 
existence of a (unique) ring homomorphism $\Z[1/30] \to \F_p$ induces morphisms of partial 
Boolean algebras $\Proj(M_3(\Z[1/30])) \subseteq \Idpt(\M_3(\Z[1/30]) \to \Idpt(\M_3(\F_p))$. 
It follows from Theorem~\ref{thm:Bub} and functoriality of Kochen-Specker colorings that 
$\Idpt(\M_3(\F_p))$ has no Kochen-Specker colorings. On the other hand, for $p \in \{2,5\}$ 
the result follows directly from Theorem~\ref{thm:counting argument}.
But as we have already mentioned, a far stronger conclusion will be made in
Corollary~\ref{cor:any ring} below.


\separate

We will now define a set $\setS \subseteq \Idpt(\M_3(\Z))$ that will be shown to have
no Kochen-Specker coloring. 
For a commutative ring $R$, given a basis $\{u,v,w\}$ of the free $R$-module $R^3$, we use the notation 
$\left[\begin{array}{@{}c|cc@{}} u & v & w \end{array} \right]$ to denote the the idempotent in
$\M_3(R)$ with range spanned by $u$ and with kernel spanned by $v$ and $w$. (This may be explicitly
constructed via the invertible matrix $U = \begin{pmatrix} u & v & w \end{pmatrix}$ as
$U E_{11} U^{-1}$.)
The argument given in the proof of Theorem~\ref{thm:counting argument}
shows that number of rank-1 idempotents in $\M_3(\F_2)$ is equal to $2^2(2^2 + 2 + 1) = 28$.
Below we list all~28 idempotent matrices of rank~1 in $\M_3(\F_2)$ in terms of the above notation.
\begingroup
\allowdisplaybreaks
\addtolength{\jot}{1em}
\begin{align*}
P_1 &= \left[ \begin{array}{@{}c|cc@{}} 
1 & 0 & 0 \\
0 & 1 & 0 \\
0 & 0 & 1
\end{array} \right]
& P_2 &= \left[ \begin{array}{@{}c|cc@{}} 
0 & 1 & 0 \\
1 & 0 & 0 \\
0 & 0 & 1
\end{array} \right]
& P_3 &= \left[ \begin{array}{@{}c|cc@{}} 
0 & 1 & 0 \\
0 & 0 & 1 \\
1 & 0 & 0
\end{array} \right]
& P_4 &= \left[ \begin{array}{@{}c|cc@{}} 
0 & 1 & 0 \\
1 & 0 & 1 \\
0 & 0 & 1
\end{array} \right] \\
P_5 &= \left[ \begin{array}{@{}c|cc@{}} 
0 & 1 & 0 \\
1 & 0 & 1 \\
1 & 0 & 0
\end{array} \right]
& P_6 &= \left[ \begin{array}{@{}c|cc@{}} 
0 & 1 & 0 \\
0 & 0 & 1 \\
1 & 0 & 1
\end{array} \right]
& P_7 &= \left[ \begin{array}{@{}c|cc@{}} 
0 & 1 & 0 \\
1 & 0 & 0 \\
1 & 0 & 1
\end{array} \right]
& P_8 &= \left[ \begin{array}{@{}c|cc@{}} 
1 & 0 & 1 \\
0 & 1 & 0 \\
0 & 0 & 1
\end{array} \right] \\
P_9 &= \left[ \begin{array}{@{}c|cc@{}} 
1 & 0 & 1 \\
0 & 1 & 0 \\
1 & 0 & 0
\end{array} \right]
& P_{10} &= \left[ \begin{array}{@{}c|cc@{}} 
0 & 0 & 1 \\
0 & 1 & 0 \\
1 & 0 & 1
\end{array} \right]
& P_{11} &= \left[ \begin{array}{@{}c|cc@{}} 
1 & 0 & 0 \\
0 & 1 & 0 \\
1 & 0 & 1
\end{array} \right]
& P_{12} &= \left[ \begin{array}{@{}c|cc@{}} 
1 & 0 & 1 \\
0 & 0 & 1 \\
0 & 1 & 0
\end{array} \right] \\
P_{13} &= \left[ \begin{array}{@{}c|cc@{}} 
1 & 0 & 1 \\
1 & 0 & 0 \\
0 & 1 & 0
\end{array} \right]
& P_{14} &= \left[ \begin{array}{@{}c|cc@{}} 
0 & 0 & 1 \\
1 & 0 & 1 \\
0 & 1 & 0
\end{array} \right]
& P_{15} &= \left[ \begin{array}{@{}c|cc@{}} 
1 & 0 & 0 \\
1 & 0 & 1 \\
0 & 1 & 0
\end{array} \right]
& P_{16} &= \left[ \begin{array}{@{}c|cc@{}} 
1 & 0 & 1 \\
1 & 1 & 0 \\
1 & 0 & 0
\end{array} \right] \\
P_{17} &= \left[ \begin{array}{@{}c|cc@{}} 
1 & 0 & 0 \\
1 & 1 & 0 \\
1 & 0 & 1
\end{array} \right]
& P_{18} &= \left[ \begin{array}{@{}c|cc@{}} 
0 & 0 & 1 \\
1 & 1 & 0 \\
1 & 0 & 1
\end{array} \right]
& P_{19} &= \left[ \begin{array}{@{}c|cc@{}} 
1 & 0 & 1 \\
1 & 1 & 0 \\
0 & 1 & 0
\end{array} \right]
& P_{20} &= \left[ \begin{array}{@{}c|cc@{}} 
1 & 1 & 0 \\
0 & 1 & 1 \\
0 & 0 & 1
\end{array} \right] \\
P_{21} &= \left[ \begin{array}{@{}c|cc@{}} 
0 & 1 & 0 \\
1 & 1 & 1 \\
0 & 0 & 1
\end{array} \right]
& P_{22} &= \left[ \begin{array}{@{}c|cc@{}} 
0 & 0 & 1 \\
1 & 0 & 1 \\
1 & 1 & 0
\end{array} \right]
& P_{23} &= \left[ \begin{array}{@{}c|cc@{}} 
1 & 0 & 1 \\
1 & 0 & 0 \\
1 & 1 & 0
\end{array} \right]
& P_{24} &= \left[ \begin{array}{@{}c|cc@{}} 
1 & 0 & 1 \\
0 & 1 & 0 \\
1 & 1 & 0
\end{array} \right] \\
P_{25} &= \left[ \begin{array}{@{}c|cc@{}} 
0 & 1 & 0 \\
0 & 1 & 1 \\
1 & 0 & 1
\end{array} \right]
& P_{26} &= \left[ \begin{array}{@{}c|cc@{}} 
1 & 1 & 0ß† \\
1 & 0 & 1 \\
0 & 1 & 0
\end{array} \right]
& P_{27} &= \left[ \begin{array}{@{}c|cc@{}} 
1 & 0 & 1 \\
0 & 0 & 1 \\
1 & 1 & 0
\end{array} \right]
& P_{28} &= \left[ \begin{array}{@{}c|cc@{}} 
1 & 1 & 0 \\
1 & 1 & 1 \\
1 & 0 & 1
\end{array} \right]
\end{align*}
\endgroup
Considering the column vectors above as elements of $\Z^3$, note that each triple of vectors given 
above forms a basis of the free $\Z$-module $\Z^3$. (This can be easily verified, for instance, by 
noting that the matrix $U = \begin{pmatrix} u & v & w \end{pmatrix}$ formed by placing the three 
vectors from a triple into its columns has determinant in the group  of units $\{\pm 1\} \subseteq \Z$.) Thus if we 
interpret the above notation in $\M_3(\Z)$, the $P_i$ for $1 \leq i \leq 28$ define distinct rank-1 
idempotents in $\M_3(\Z)$. We let $\setS = \{P_i \mid 1 \leq i \leq 28\}$ denote the set of~$28$ 
idempotents in $\M_3(\Z)$ given above.

Note that two idempotents $P =\left[\begin{array}{@{}c|cc@{}} u & v & w \end{array} \right]$ and
$Q = \left[\begin{array}{@{}c|cc@{}} x & y & z \end{array} \right]$ presented as above are orthogonal
if and only if the range vector $u$ is contained in the $\Z$-span of the kernel vectors $y$ and $z$,
and also $x$ is contained in the $\Z$-span of $v$ and $w$.

Under the canonical ring homomorphism $ \phi \colon \M_3(\Z) \to \M_3(\F_2)$ induced entrywise by 
$\Z \twoheadrightarrow \F_2$, the set 
$\setS$ is constructed in such a way that $\phi$ restricts to a bijection 
$\setS \overset{\sim}{\longrightarrow} \Idpt(\M_3(\F_2))$. In this sense, we may think of $\setS$ as a
set of ``lifts'' of the idempotents of $\M_3(\F_2)$ to the integer $3 \times 3$ matrices.

If the idempotents in $\setS$ satisfied the same orthogonality relations as their images in $\M_3(\F_2)$,
then it would follow directly from Theorem~\ref{thm:counting argument} that $\setS$ has no 
Kochen-Specker coloring. But as it happens, there are triples of idempotents in $\setS$ that
are not pairwise orthogonal over $\Z$, but whose image under $\phi$ become pairwise orthogonal 
over $\F_2$. 

To be precise, the following list displays all triples $\{i,j,k\}$ such that $\{P_i, P_j, P_k\} \subseteq 
\M_3(\Z)$ and $\{\phi(P_i), \phi(P_j), \phi(P_k)\} \subseteq \M_3(\F_2)$ are both orthogonal.
\begingroup
\allowdisplaybreaks
\begin{align*}
\O_1 &= \{1,2,3\} & \O_2 & = \{1,4,5\} & \O_3 &= \{1,6,7\} & \O_4 &= \{2,8,9\} \\
\O_5 &= \{2,10,11\} & \O_6 &= \{3, 12, 13\} & \O_7 &= \{3, 14, 15\} & \O_8 &= \{4, 8, 16\} \\
\O_9 &= \{4,17,18\} & \O_{10} &= \{5,19,20\} & \O_{11} &= \{5,15,21\} & \O_{12} &= \{6,17,22\} \\
\O_{13} &= \{6,12,23\} & \O_{14} &= \{8,23,24\} & \O_{15} &= \{10,14,17\} & \O_{16} &= \{10,23,27\} \\
\O_{17} &= \{12, 16, 19\} & \O_{18} &= \{15, 22, 25\} & \O_{19} &= \{14,16,26\} & \O_{20} &= \{19,22,28\}
\end{align*}
\endgroup
On the other hand, the list below displays the pairs $\O_{m}^a = \{i,j\}$ and $\O_{m}^b = 
\{i,k\}$ of indices such that $\{P_i,P_j\}$ and $\{P_j, P_k\}$ are orthogonal pairs with $P_i$ and $P_k$ 
not orthogonal over $\Z$, but $\{\phi(P_i), \phi(P_j), \phi(P_k)\}$ forms an orthogonal triple over 
$\F_2$.
\begingroup
\allowdisplaybreaks
\begin{align*}
\O_{21}^a &= \{7,24\} & \O_{21}^b &= \{7, 20\} & \O_{22}^a &= \{7,11\} & \O_{22}^b &= \{7,25\} \\
\O_{23}^a &= \{9,26\} & \O_{23}^b &= \{21,26\} & \O_{24}^a &= \{9,13\} & \O_{24} &= \{13,20\} \\
\O_{25}^a &= \{11,18\} & \O_{25}^b &= \{18,21\} & \O_{26}^a &= \{13, 27\} & \O_{26}^b &= \{13,25\} \\
\O_{27}^a &= \{18, 28\} & \O_{27}^b &= \{18,24\} & \O_{28}^a &= \{26, 28\} & \O_{28}^b &= \{26, 27\}
\end{align*}
\endgroup

Thus the proof that $\setS$ is uncolorable requires a different argument. Note that the following
proof only gives a complete argument that a larger set $\setS' \supseteq \setS$ is uncolorable.
But as we indicate below, with extra work it is possible to prove that $\setS$ itself is uncolorable.

\begin{theorem}\label{thm:integer KS}
There is no Kochen-Specker coloring of $\Idpt(\M_3(\Z))$.
\end{theorem}

\begin{proof}
It suffices to show that the set $\setS$ defined above is uncolorable. This is achieved through
a case-splitting argument which shows that various colorings of the triples $\O_i$ for $i = 1,2,3$
lead to contradictions. Thus, we assume toward a contradiction that $\setS$ has a Kochen-Specker
coloring.

\textbf{Case I:} \emph{$P_1$ black, $P_2$ black, and $P_3$ white in $\O_1$.}
Then in the orthogonal triple $\O_6$ we have that $P_{12}$ is black. We examine two possible
subcases.

\textbf{Case I.1:} \emph{$P_4$ is black and $P_5$ is white in $\O_2$.}
We obtain the following sequence of deductions:
\begin{align*}
P_{19} \mbox{ black in triple } \O_{10} &\Rightarrow P_{16} \mbox{ white in triple } \O_{17} \\
&\Rightarrow P_8 \mbox{ white in triple } \O_8 \\
\end{align*}
We deduce contradictions in two further subcases of Case~I.1 as follows.

\textbf{Case I.1.a:} \emph{$P_6$ black and $P_7$ white in $\O_3$.} 
In this case we deduce:
\begin{align*}
P_{24} \mbox{ black in pair } \O_{21}^a &\Rightarrow P_{11} \mbox{ black in pair } \O_{22}^a \\
&\Rightarrow P_{10} \mbox{ black in triple } \O_5 \\
&\Rightarrow P_{23} \mbox{ black in triple } \O_{14},
\end{align*}
arriving at the contradiction that two idempotents in triple $\O_{16}$ are white. 

\textbf{Case I.1.b:} \emph{$P_6$ black and $P_7$ white in $\O_3$.} 
Now we deduce that $P_{17}$ and $P_{22}$ are black in triple $\O_{12}$, which implies that  
$P_{18}$ is white in triple $\O_9$ and $P_{28}$ is white in triple $\O_{20}$. We obtain a 
contradictory coloring of the pair $\O_{28}^a$. This completes the proof that case~I.1 leads to 
a contradiction.

\textbf{Case I.2:} \emph{$P_4$ is white and $P_5$ is black in $\O_2$.} 
In this setting we have $P_{14}$ black in triple $\O_7$, while $P_8$ and $P_{16}$ are black in 
triple $\O_8$. It follows that $P_9$ is white in triple $\O_4$ and $P_{26}$ is white in triple 
$\O_{19}$. We obtain a contradictory coloring of the pair $\O_{23}^a$. 

We deduce from these contradictions that there is no Kochen-Specker coloring of $\setS$ satisfying
the condition in case~I. Let $G \subseteq \M_3(\Z)$ be the group of permutation matrices, acting on 
$\Idpt(\M_3(\Z))$ by conjugation. While $\setS$ is not fixed under the action of $G$, it is contained in 
a smallest set $\setS' = G\setS$ closed under this action. Now if $\setS'$ had a Kochen-Specker 
coloring, conjugation by a suitable element of $G$ would yield a coloring of $\setS' \supseteq \setS$ 
for which case~I indeed holds. So we find that $\setS'$ is  uncolorable, proving the theorem.

(The interested reader may reason as above to verify similar contradictions in either
\textbf{case~II:} $P_1$ black, $P_2$ white, $P_3$ black, or \textbf{case~III:} $P_1$ 
white, $P_2$ black, $P_3$ black. This proves that $\setS$ itself has no Kochen-Specker coloring.)
\end{proof}

\begin{corollary}\label{cor:any ring}
Let $R$ be a ring, and fix any integer $n \geq 3$.
\begin{enumerate}
\item There is no Kochen-Specker coloring of $\Idpt(\M_n(R))$. 
\item $\pSpec(\M_n(R)) = \varnothing$. 
\item There is no morphism of partial rings from $\M_n(R)$ to any (total) commutative ring.
\item The colimit in $\cRing$ of the diagram of commutative subrings of $\M_n(R)$ is zero.
\end{enumerate}
\end{corollary}

\begin{proof}
(1) It follows from Lemma~\ref{lem:obstruction passes on} and Theorem~\ref{thm:integer KS} that 
$\Idpt(\M_n(\Z))$ has no Kochen-Specker coloring. The existence of a ring homomorphism
$\M_n(\Z) \to \M_n(R)$ and functoriality of $\Idpt$ yield a morphism of partial Boolean algebras
$\Idpt(\M_n(\Z)) \to \Idpt(\M_n(R))$. Now by functoriality of Kochen-Specker colorings, we obtain
a function $\KS(\Idpt(\M_n(R))) \to \KS(\Idpt(\M_n(\Z))) = \varnothing$. Now we deduce that 
$\Idpt(\M_n(R))$ has no Kochen-Specker colorings.   

Part~(2) follows from~(1) by Corollary~\ref{cor:coloring to spectrum}. Then~(3) and~(4) are 
immediate from Lemma~\ref{lem:empty spec obstructions}.
\end{proof}

We close this section with some open questions related to ``ring-theoretic contextuality''
as it is discussed in Section~\ref{sec:intro}. To date, the only proofs that a noncommutative 
ring has an empty partial spectrum rely upon Kochen-Specker uncolorability of idempotents, 
as in Corollary~\ref{cor:coloring to spectrum} above. Thus it would be interesting to find a ring
with empty partial spectrum in the extreme case with only the trivial idempotents $0$ and $1$.

\begin{question}
Does there exist a nonzero ring $R$ with no nontrivial idempotents such that 
$\pSpec(R) = \varnothing$?
\end{question}

By Lemma~\ref{lem:empty spec obstructions}, an example of such a ring will have no morphism
of partial rings $R \to C$ for any nonzero commutative ring $C$, and thus will have no 
``noncontextual hidden variable theory.'' If $R$ is a domain (that is, a nonzero ring without zero divisors),
then the zero ideal is readily seen to be a prime partial ideal so that $\pSpec(R)$ is nonempty;
nevertheless, it would be interesting to find an example of such $R$ that still has no
``noncontextual hidden variable theory.'' Thus we ask the following.

\begin{question}
Does there exist a domain $R$ such that there is no morphism of partial rings
$R \to C$ for any nonzero commutative ring $C$?
\end{question} 

\section{Applications to spectrum functors in noncommutative algebraic geometry}
\label{sec:spectrum}

In this final section, we apply the above results on the partial spectrum of integer matrix rings
to strengthen the main result of~\cite{Reyes} and certain results from~\cite{BergHeunen:extending}.

Modern algebraic geometry provides a way to view \emph{every} commutative ring
as a ring of ```globally defined functions'' (the global sections of a sheaf of rings) on a 
geometric object (a locally ringed space) called an \emph{affine scheme}. 
The scheme associated to a commutative ring is called its
\emph{spectrum}, and the assignment of the spectrum to each ring forms an equivalence of
categories $\cRingop \to \AffSch$.
For a commutative ring $R$, the Zariski prime spectrum $\Spec(R)$ (the set of prime ideals of
$R$) forms the underlying set of its affine scheme. 
We refer readers to~\cite[I.2]{EisenbudHarris} for an introduction to the spectrum of a ring
in algebraic geometry.

In the spirit of noncommutative geometry, it is natural to wonder whether every noncommutative
ring may be given a similar ``spatial realization.'' The most obvious way to attempt to build a
``noncommutative affine scheme'' would be to use a ringed space for which the sheaf
of rings is not necessarily commutative. Indeed, such constructions have been intensely pursued
in past decades; an outstanding survey of these efforts may be found in~\cite{VanOystaeyen}.
In order to obtain a true correspondence between algebra and geometry, one would wish
for such a construction to be a contravariant functor. Thus, at the very least, one would
require a functor $F$ from $\Ringop$ to the category $\Top$ of topological spaces, or
even to $\Set$ if we forget about topology, that yields the underlying point set of each
ringed space, such that the restriction of $F$ to $\cRingop$ is (isomorphic to) the usual
spectrum functor $\Spec \colon \cRingop \to \Set$. Furthermore, in order to obtain a
nontrivial construction, one should require that if $R$ is a nonzero ring then $F(R)$ is
nonempty. 

However, it was shown in~\cite{Reyes} that any functor $F$ as above necessarily assigns 
$F(\M_n(R)) = \varnothing$ for any ring $R$ containing $\C$ as a subring and any integer
$n \geq 3$. The proof of this result crucially relied upon the fact that the Kochen-Specker 
Theorem implies that $\pSpec(\M_n(R))= \varnothing$ for any such $R$.
Our algebraic analogues of Kochen-Specker will allow us to extend this result to any ring
$R$, not only those that contain the complex numbers.

In hindsight, the connection between the connection between this algebro-geometric
obstruction and the Kochen-Specker Theorem is arguably a natural one. For
a ring $R$, let $\catC(R)$ denote the diagram in $\cRing$ whose objects are the 
commutative subrings $C \subseteq R$ and whose morphisms $C_1 \to C_2$ are the
inclusions of subrings $C_1 \subseteq C_2 \subseteq R$. It is shown in~\cite[Proposition~2.14]{Reyes} that
as sets, the partial spectrum of $R$ is the limit in the category of sets of the spectra
$\Spec(C)$ for $C \in \catC(R)$:
\begin{equation}\label{eq:Spec limit}
\pSpec(R) \cong \varprojlim_{C \in \catC(R)} \Spec(C).
\end{equation}
This bijection allows us to view a point in $\pSpec(R)$ as a ``noncontextual choice
of points'' in the spectra $\Spec(C)$ of all commutative subrings of $R$. The
obstruction of~\cite{Reyes} used the Kochen-Specker Theorem to show that
there does not exist any such ``noncontextual choice of points'' in the case where
$R = \M_n(\C)$ for $n \geq 3$. See also ~\cite[Sec.~3--4]{Doring} for a related 
discussion.

Further, it is interesting to note that Kochen and Specker's motivating discussion 
in~\cite[Section~1]{KochenSpecker} phrases the problem of hidden variables as the search for
a probability space $\Omega$ of ``hidden pure states,'' with a hidden variable theory being
a morphism of partial algebras from the algebra of observables to the algebra of real-valued
measurable functions on $\Omega$. If one imagines $\Omega$ as a kind of spectrum associated
to a quantum system, then it seems entirely natural that the Kochen-Specker theorem should have 
led to the results of~\cite{Reyes}.

We now prove the strengthened version of~\cite[Theorem~1.1]{Reyes}, 
answering the question posed in~\cite[Question~4.2]{Reyes}. We use
essentially the same argument, relying upon Theorem~\ref{thm:integer KS}
in place of the Kochen-Specker Theorem. 

\begin{theorem}\label{thm:empty Spec}
Let $F \colon \Ringop \to \Set$ (or $F \colon \Ringop \to \Top$) be a functor whose 
restriction to the full subcategory $\cRingop$ is isomorphic to $\Spec$. Then 
$F(\M_n(R)) = \varnothing$ for any ring $R$ and any integer $n \geq 3$.
\end{theorem}

\begin{proof}
By~\cite[Theorem~2.15]{Reyes}, the hypothesis on $F$ ensures that there exists a natural 
transformation $F \to \pSpec$. 
For a fixed ring $R$ and integer $n \geq 3$, the unique ring homomorphism $\Z \to R$ induces a ring
homomorphism $\M_n(\Z) \to \M_n(R)$. The natural transformation and functoriality of $\pSpec$
yield a composite function
\[
F(\M_n(R)) \to \pSpec(\M_n(R)) \to \pSpec(\M_n(\Z)) = \varnothing,
\]
with the last equality following from Corollary~\ref{cor:any ring}. As the only set with a function
to the empty set is $\varnothing$ itself, we conclude that $F(\M_n(R)) = \varnothing$.
\end{proof}

This obstruction to ``noncommutative spectrum functors'' has the following immediate application,
which strengthens~\cite[Corollary~4.3]{Reyes} regarding ``abelianization functors''
defined on the category of rings.

\begin{corollary}\label{cor:abelianization}
Let $R$ be a ring and $n \geq 3$ be an integer. 
If $\alpha \colon \Ring \to \cRing$ is any functor whose restriction to $\cRing$ is isomorphic
to the identity functor, then $\alpha(\M_n(R)) = 0$.
\end{corollary}

\begin{proof}
The hypothesis on $\alpha$ ensures that the composite functor $\Spec \circ \alpha \colon
\Ringop \to \Set$ has restriction to $\cRingop$ isomorphic to $\Spec$. By Theorem~\ref{thm:empty Spec}
we have $\Spec(\alpha(\M_3(\Z)) = \varnothing$. This implies that the commutative ring
$\alpha(\M_3(\Z))$ is zero.
\end{proof}

It was noted in~\cite[p.~689]{Reyes} that the statements of both Corollaries~\ref{cor:any ring}(4)
and~\ref{cor:abelianization} fail in the case where $n = 2$.

\separate

The study of topological spaces and their sheaves, especially including ringed spaces, can be conducted
without any reference to the underlying point set of the topological space via the use of \emph{locales}
and \emph{toposes}~\cite{Johnstone:pointless}. One might therefore expect that obstructions such as 
the one in Theorem~\ref{thm:empty Spec} could be avoided if one considers the Zariski spectrum as a 
functor taking values in the category of locales rather than in the category of topological spaces. 
(We will describe the localic Zariski spectrum below.)
However, it was shown in~\cite{BergHeunen:extending} that
the obstruction of~\cite{Reyes} persists for functors taking values in such categories of ``pointless 
spaces.'' 
We now show that our version of the Kochen-Specker theorem for integer matrices allows us to extend 
the obstruction of Theorem~\ref{thm:empty Spec} in the same manner. 
From this point until the end of the paper, we now consider the Zariski spectrum $\Spec(R)$ of a
commutative ring $R$ as a topological space, with the usual Zariski topology whose open sets are
those of the form $D(I) = \{\p \in \Spec(R) \mid I \nsubseteq \p\}$ for all ideals $I$ of $R$.
%
%
We include the basic definitions of locales and the localic spectrum below, although we must invoke some results
from locale theory in the proof below without a full survey of this theory, which would take us too far afield.

A \emph{frame} $(F,\bigvee,\wedge,0,1)$ is a complete lattice which satisfies the ``infinite distributive
law'' $a \wedge \left( \bigvee b_i \right) = \bigvee (a \wedge b_i)$ for any family $\{b_i\} \subseteq F$. 
The motivating example of a frame is the collection of open sets in a topological space $X$. Frames form 
a category $\Frm$ whose morphisms are the homomorphisms of posets that preserve finite meets and 
arbitrary joins; in particular, these maps preserve $0$ and $1$.
The category $\Loc = \Frm\op$ of \emph{locales} is defined to be the opposite of the category of frames.
If $L$ denotes a locale, we will use $\Omega(L)$ to denote its underlying frame (so that $L$ is
``opposite'' to $\Omega(L)$ in $\Loc = \Frm\op$); we call the elements of $\Omega(L)$ the 
\emph{opens} of $L$. Given a morphism $f \colon L \to S$ in $\Loc$, we denote the corresponding
morphism of frames as $f^* \colon \Omega(S) \to \Omega(L)$.

We recall one formulation of the localic Zariski spectrum from~\cite[V.3]{Johnstone:Stone} 
(and especially Corollary~V.3.2(i) of that reference). 
For a commutative ring $R$ and an ideal $I$ of $R$, recall that the \emph{radical of $I$} is the ideal
$\sqrt{I} = \{x \in R \mid x^n \in I \mbox{ for some integer } n \geq 1\}$, and that $I$ is
called a \emph{radical ideal} if $I = \sqrt{I}$ (that is, $x^n \in I$ for $x \in R$ and some integer
$n \geq 1$ implies $x \in I$). Let $\RIdl(R)$ denote the set of radical ideals of $R$. 
This forms a lattice with respect to inclusion, which is complete since the intersection of an arbitrary
set of radical ideals is again radical. The join of an arbitrary family $\{I_j\} \subseteq \RIdl(R)$
and the pairwise meet of $I,J \in \RIdl(R)$ are given in terms of the usual ideal sum and product by
\[
\bigvee I_j = \sqrt{\sum I_j}, \qquad I \wedge J = I \cap J = \sqrt{I \cdot J}.
\]
With these descriptions, one can verify that the ``infinite distributive law''
\[
J \wedge \left( \bigvee I_j \right) = \sqrt{ J \cdot \sqrt{\sum I_j} } = \sqrt{ \sum JI_j } 
= \sqrt{ \sum \sqrt{JI_j} } = \bigvee (J \wedge I_j)
\]
holds for all $J$ and $\{I_j\}$ as above. Thus $\RIdl(R)$ is a frame. 

We define the \emph{localic (Zariski) spectrum} of $R$ to be the locale $\LSpec(R)$ whose 
corresponding frame is $\Omega(\LSpec(R)) = \RIdl(R)$. A morphism $f \colon R \to S$ in
$\cRing$ induces a function $\RIdl(f) \colon \RIdl(R) \to \RIdl(S)$ via
$I \mapsto \sqrt{S \cdot f(I)}$, which one may verify to be a morphism of frames;
we denote the corresponding morphism of locales by $\LSpec(f) \colon \LSpec(S) \to \LSpec(R)$. 
In this way, the localic spectrum forms a functor
\[
\LSpec \colon \Ringop \to \Loc.
\]

Readers familiar with the (``spatial'') Zariski spectrum will recognize that 
$\RIdl(R)$ is isomorphic to the lattice of open sets of the Zariski topology on 
$\Spec(R)$ (see~\cite[Lemma~2.1]{Hartshorne}, for instance). Indeed, the
spatial Zariski spectrum $\Spec(R)$ is isomorphic to the \emph{space of points}~\cite[II.1.3]{Johnstone:Stone} of the locale $\LSpec(R)$;
see~\cite[V.3.2]{Johnstone:Stone}. However, 
the definition of $\LSpec$ is entirely constructive, while one must invoke the
Axiom of Choice (or at least the Boolean Prime Ideal Theorem) to verify that
$\Spec(R)$ is nonempty for nonzero rings $R$. For this reason, localic spectra
are preferred in the setting of constructive mathematics.

We say that a locale is \emph{trivial} if its frame of opens is a singleton (i.e., satisfies 
$0 = 1$), or equivalently, if it is an initial object of $\Loc$. 
Trivial locales play the role of the empty space in pointless topology.
If $L$ is a locale with a morphism to a trivial locale, then the top and bottom elements 
of the frame of opens of $L$ are also equal, making $L$ a trivial locale. 

\begin{corollary}\label{cor:localic Spec}
Let $F \colon \Ringop \to \Loc$ be a functor whose restriction to $\cRingop$ is isomorphic
to $\LSpec \colon \cRingop \to \Loc$. Then $F(\M_n(R))$ is the trivial locale for every ring $R$ and every
integer $n \geq 3$.
\end{corollary}

\begin{proof}
Let $G \colon \Ringop \to \Loc$ be the functor that assigns to a ring $R$ the limit of the
locales $\LSpec(C)$ where $C$ ranges over the diagram $\catC(R)$ of commutative subrings of $R$. It is clear from the
construction of $G$ that for any functor $F$ as in the statement, there is a natural transformation 
$G \to F$, as in~\cite[Theorem~2.15]{Reyes}.
The category of \emph{coherent locales} is the subcategory of $\Loc$ that is opposite to the 
essential image of the functor that assigns to each distributive lattice its frame of ideals; 
see~\cite[II.3]{Johnstone:Stone}.
It is known that $\LSpec$ has image in the category of coherent locales~\cite[V.3.1]{Johnstone:Stone}.
Furthermore, it is known~\cite[Lemma~2.6]{BergHeunen:extending} that the subcategory of 
coherent locales is closed under limits in $\Loc$.

Let $*$ denote the locale corresponding to the one-point space, so that the point-set functor
$\pt \colon \Loc \to \Set$ is $\pt = \Loc(*,-)$. Being representable, this functor preserves 
limits~\cite[Theorem~V.4.1]{MacLane}. 
Thus for any ring $R$ we have natural isomorphisms
\[
\pt(G(R)) = \pt\left( \varprojlim_{C \in \catC(R)} \LSpec(C) \right) \cong \varprojlim_{C \in \catC(R)} \pt(\LSpec(C)).
\]
As discussed above, the composite $\pt \circ \Spec \colon \cRingop \to \Set$ is isomorphic to the usual 
Zariski prime spectrum functor $\Spec$.
It follows from~\eqref{eq:Spec limit} that the limit above is naturally isomorphic to $\pSpec(R)$, 
so we in fact have a natural isomorphism $\pt \circ \, G \cong \pSpec$. 

Now for any ring $R$ and any integer $n \geq 3$, we have $\pt(G(\M_n(R)) \cong \pSpec(\M_n(R)) = 
\varnothing$ thanks to Theorem~\ref{thm:empty Spec}. Because $G(\M_n(R))$ is coherent, it is a 
\emph{spatial} locale~\cite[II.3.4]{Johnstone:Stone}, so that its frame of opens is isomorphic to the
frame of open subsets of its space of points. Thus $G(\M_n(R))$ is the trivial locale. Finally,
the existence of a morphism of locales $G(\M_n(R)) \to F(\M_n(R))$ implies that $F(\M_n(R))$ 
is also trivial.
\end{proof}

We also remark that as in~\cite{BergHeunen:extending}, similar obstructions hold if we regard 
$\Spec$ as a functor from commutative rings into the any one of the categories of 
toposes, ringed toposes, ringed locales, or ringed spaces. 

\begin{corollary}
Let $\catC$ be any of the categories 
of toposes, ringed toposes, 
ringed locales, or ringed spaces, and consider $\Spec$ as a functor  $\cRingop \to \catC$ in the 
usual way. 
Suppose that $F \colon \Ringop \to \catC$ is a functor whose restriction to $\cRingop$ is 
isomorphic to $\Spec$. Then $F(\M_n(R))$ is the trivial (initial) object of $\catC$ for any ring 
$R$ and any integer $n \geq 3$.
\end{corollary}

\begin{proof}
%
The proofs 
are direct analogues of those given in~\cite[Corollaries~6.2 and~6.3 ]{BergHeunen:extending}.
\end{proof}

\begin{remark}
Given the emphasis in the locale theory literature on constructive proofs, we wish to emphasize 
that the proof of Corollary~\ref{cor:localic Spec} above invokes the nonconstructive technique of 
reducing to spaces of points; by contrast, the proof for the case where $R = \C$ given 
in~\cite[Corollary~6.1]{BergHeunen:extending} is constructive. 
The proofs of these localic obstructions are complicated by the fact that the Zariski spectrum (either spatial or localic) generally 
does not preserve limits out of $\Ringop$. This was handled in~\cite{BergHeunen:extending} 
by noting that the Zariski spectrum \emph{does} preserve limits when restricted to 
finite-dimensional algebras over the algebraically closed field $\C$. 
Unfortunately, more work is required in our context because our algebras are not finite-dimensional
over a field. There is a proof of Corollary~\ref{cor:localic Spec} that is constructive in principle.
The idea is to consider the localic spectra of the finite diagram of the commutative subrings of $\M_3(\Z)$ 
that are generated by orthogonal sets of idempotents from the proof of Theorem~\ref{thm:integer KS}, 
to show that $\LSpec$ restricted to this diagram (whose limit is trivial) has image in the subcategory 
$\fBool\op$ opposite to finite Boolean algebras within $\Loc = \Frm\op$, and to proceed as 
in~\cite[Corollary~6.1]{BergHeunen:extending} noting that the inclusion $\fBool\op \hookrightarrow \Loc$ 
preserves limits.
With such techniques in hand, one may also extend the proof of~\cite[Corollary~5.7]{BergHeunen:extending}
to show that a similar obstruction holds for extensions of $\Spec$ into the opposite of the category 
of (unital or strong) quantales. 
For the sake of brevity, we do not provide further details here.
\end{remark}

In closing, we note that readers seeking further open problems regarding contextuality 
in noncommutative algebraic geometry will find some in~\cite[Question~4.9]{Reyes2}, 
while readers interested in positive results on noncommutative spectrum functors are 
referred to~\cite{HeunenReyes:active} for one successful example.

\appendix
\renewcommand{\thesection}{A}
\setcounter{section}{0}

\section[Partial algebra morphisms from Kochen-Specker colorings]{Partial algebra morphisms from Kochen-Specker colorings,\\ by Alexandru Chirvasitu}

We prove here the following result relating idempotent colorings to morphisms of partial algebras.

\begin{lemma}\label{lem:coloring to morphism}
  Let $F$ be a perfect field, $K\supseteq F$ a field extension
  containing an isomorphic copy of every degree-two and degree-three
  extension of $F$, and $A$ a partial $F$-subalgebra of
  $\M_3=\M_3(F)$.

  If there exists a Kochen-Specker coloring of $\Idpt(A)$, then there
  exists a morphism of partial $F$-algebras $A \to K$. Consequently,
  $\pSpec(R)$ is nonempty.
\end{lemma}
\begin{proof}
  A Kochen-Specker coloring provides a map $\varphi$ from $\Idpt(A)$
  to $\{0,1\}$ compatible with addition of orthogonal idempotents, and
  we wish to extend this map to all of $A$.

  {\bf Step 1: Reducing to semisimple operators.} The fact that $F$ is
  perfect ensures that we can decompose every $x\in A$ as a sum
  $x_s+x_n$, where $x_s\in A$ is semisimple, $x_n\in A$ is nilpotent,
  and each is a polynomial in $x$ with no constant term. Because
  $x,y\in A$ commute if and only if $x_s$ and $x_n$ both commute with
  $y_s$ and $y_n$, we can simply extend $\varphi$ to the partial
  algebra generated by idempotents and nilpotent operators by sending
  the latter to zero.

  If we had an extension of $\varphi$ to the partial subalgebra
  $A_{ss}\subset A$ consisting of semisimple elements, then we could
  set $\varphi(x)=\varphi(x_n)+\varphi(x_s)$. The above observation
  that $x,y\in A$ commute if and only if $x_s$ and $x_n$ commute with
  $y_s$ and $y_n$ then ensures that this is well defined and a partial
  algebra morphism.

  We may now assume that all elements of $A$ are semisimple; this
  assumption will be in place throughout the rest of the proof.

  The {\it support} $\mathrm{supp}(x)\in\Idpt(\M_3)$ of a semisimple
  element $x\in\M_3$ is the idempotent with the same range and kernel
  as $x$. For every $x\in A$ consider the element $x_d$ (for
  `diagonalizable') defined as $\sum t_ip_i$, where $t_i$ are the
  distinct non-zero eigenvalues of $x$ and $p_i$ are the corresponding
  spectral idempotents.

  The element $x_d$ is expressible as a polynomial in $x$ with no
  constant term (because $p_i$ are so expressible) and hence belongs
  to $A$. Moreover, it is the unique element of $\M_3$ that is
  diagonalizable in $\M_3$, a polynomial in $x$ with no constant term,
  and whose support is maximal among elements with this property.

  It follows from the description of $x$ that $x-x_d$ is either zero
  or {\it purely non-diagonalizable}, in the sense that $(x-x_d)_d$
  vanishes (i.e. $x-x_d$ has no non-zero eigenvalues). Denote
  $x_{nd}=x-x_d$.

  {\bf Step 2: Diagonalizable operators.} The decomposition
$x=x_d+x_{nd}$ is similar in spirit to the Jordan decomposition, and
we can put it to similar use.

  Any diagonalizable $x\in A$ breaks up uniquely as $\sum t_ip_i$
where $t_i\in F$ and $p_i\in \Idpt(A)$ are as above. Now simply set
$\varphi(x)=\sum t_i\varphi(p_i)$. This is easily seen to be a partial
algebra morphism from the partial subalgebra $A_d\subseteq A$
consisting of diagonalizable operators to $F\subseteq K$.

{\bf Step 3: Purely non-diagonalizable operators.} Let $x\ne 0$ be
such an operator and $\langle x\rangle$ the (non-unital) subalgebra of
$A$ that it generates. It is isomorphic to a field extension of $F$
generated by any one of the non-zero eigenvalues of $x$, with unit
$\mathrm{supp}(x)$. Define an arbitrary unital algebra morphism
\begin{equation*} \psi:\langle x\rangle\to K
\end{equation*} and extend $\varphi$ to $\langle x\rangle$ via
$\varphi(x)=\psi(x)\varphi(\mathrm{supp}(x))$.

{\bf Step 4: Putting the ingredients together.} For $x\in A$ set
$\varphi(x)=\varphi(x_d)+\varphi(x_{nd})$.

{\bf Step 5: Checking that $\varphi$ is a morphism.} We have to check
that $\varphi$ as defined above preserves products and sums of
commuting elements $x,y$, which we fix throughout the rest of the
proof.

Because both $x_d$ and $x_{nd}$ can be expressed as polynomials with
no constant term in $x$, an operator commutes with $x$ if and only if
it commutes with $x_d$ and $x_{nd}$. Consequently, $x,y\in A$ commute
if and only if $x_d$, $x_{nd}$, $y_d$ and $y_{nd}$ all commute.

If $x$ and $y$ are diagonalizable there is nothing to check, as we
already know that $\varphi|_{A_d}$ is a morphism of partial
algebras. So we may as well suppose $x_{nd}\ne 0$.

Now $y$ commutes with the idempotent $e=\mathrm{supp}(x_{nd})$, and
since $\varphi$ annihilates exactly one of $e$ and $1-e$ we may as
well restrict our attention to $e\M_3e$ or $(1-e)\M_3(1-e)$, depending
on whether $\varphi(e)=1$ or $\varphi(1-e)=1$ respectively.

There are two possibilities for $e$: either it has rank two and
$\langle ex\rangle\subset e\M_3e$ is a field $L$ of degree two over
$F$, or $e=1$ and $\langle x\rangle$ is a field of degree three over
$F$. This means that if $\varphi(1-e)=1$ then $(1-e)\M_3(1-e)$ is (at
most) one-dimensional with $(1-e)x$ and $(1-e)y$ both scalars therein,
so there is nothing to check.

On the other hand, if $\varphi(e)=1$ then $ey\in e\M_3e$ commutes with
$ex$ and hence acts as $L$-linear endomorphisms of $L$. Putting $ex\in
e\M_3e$ in rational normal form will identify $eF^n$ with $eL$ and
hence $e\M_3e$ with $\mathrm{End}_F(L)$ in such a manner that $\langle
ex\rangle$ gets identified with $L\subset\mathrm{End}_F(L)$ (acting on
itself by multiplication). Since $ey\in\mathrm{End}_F(L)$ acts on $L$
as $L$-linear endomorphisms (because it commutes with $\langle
ex\rangle$) we must have $ey\in L=\langle ex\rangle$. We can now
conclude from the fact that $\varphi$ is a morphism when restricted to
$\langle x\rangle$.
\end{proof}

\begin{remark}\label{rem:finite extension}
  In particular, we can take $K=\overline{F}$, the algebraic closure
  of $F$. On the other hand, if $F$ is finite then $K$ can be taken to
  be its unique degree-six extension.
\end{remark}

\bibliographystyle{plain}
\bibliography{kochenspecker_arxiv_v3}
\end{document}